\documentclass[11pt]{article}
\usepackage{array}
\usepackage{mathtools}
\usepackage{setspace}
\usepackage[dvips]{graphicx}
\usepackage{enumerate}
\usepackage[margin=1in]{geometry}
\usepackage{amsthm,amsmath,amssymb}
\usepackage[longnamesfirst,round]{natbib}
\usepackage{natbib}
\usepackage{authblk}
\usepackage{braket}
\usepackage{xargs}
\usepackage{dsfont}
\usepackage[toc,page]{appendix}
\usepackage{color}
\usepackage{float}
\usepackage{tabularx}
\usepackage{tikz}
\usepackage{multirow}
\usepackage{float}

\allowdisplaybreaks[1]

\numberwithin{equation}{section}

\newtheorem{theorem}{Theorem}[section]
\newtheorem{lemma}[theorem]{Lemma}

\newtheorem{remark}[theorem]{Remark}

\newtheorem{example}[theorem]{Example}

\newtheorem{conjecture}[theorem]{Conjecture}

\newcommand{\var}{\mbox{Var}}
\newcommand{\D}{{\rm d}}

\newcommand{\PP}{\mathbb{P}}

\newcommand{\EE}{\mathbb{E}}

\newcommand{\Go}{\Rightarrow}

\newcommand{\RR}{\mathbb{R}}

\newcommand{\clb}{\mathcal{B}}
\newcommand{\bars}{\bar{S}}
\newcommand{\bari}{\bar{I}}

\newcommand{\hats}{\hat{S}}
\newcommand{\hati}{\hat{I}}

\newcommand{\pp}{\mathcal{P}}
\newcommand{\cla}{\mathcal{A}}

\newcommand{\clr}{\mathcal{R}}
\newcommand{\diag}{\mbox{diag}}
\newcommand{\hatn}{\hat{N}}
\newcommand{\barn}{\bar{N}}
\newcommand{\X}{{\bf \Phi}}
\newcommand{\B}{{\bf B}}

\title{\bf Approximation Methods for Analyzing Multiscale Stochastic Vector-borne Epidemic Models}
\author[1]{Xin Liu\thanks{Email: xliu9@clemson.edu}}
\author[2]{Anuj Mubayi\thanks{Email: amubayi@asu.edu}}
\author[3]{Dominik Reinhold\thanks{Email: dominik.reinhold@ucdenver.edu}}
\author[1]{Liu Zhu\thanks{Email: liuz@g.clemson.edu}}
\affil[1]{Department of Mathematical Sciences, Clemson University}
\affil[2]{School of Human Evolution and Social Change; Simon A. Levin Mathematical Computational and Modeling Science Center, Arizona State University-Tempe}
\affil[3]{Department of Biostatistics and Informatics, University of Colorado Denver}

\begin{document}

\maketitle

\begin{abstract}

Stochastic epidemic models, generally more realistic than deterministic counterparts, have often been seen too complex for rigorous mathematical analysis because of level of details it requires to comprehensively capture the dynamics of diseases. This problem further becomes intense when complexity of diseasees increases as in the case of vector-borne diseases (VBD). The VBDs are human illnesses caused by pathogens transmitted among humans by intermediate species, which are primarily arthropods.  
In this study, a stochastic VBD model, capturing demographic stochasticity and different host and vector dynamic scales, is developed and novel mathematical methods are described and evaluated to systematically analyze the model and understand its complex dynamics. The VBD model incorporates some relevant features of the VBD transmission process including demographical, ecological and social mechanisms. The analysis is based on dimensional reductions and model simplications via scaling limit theorems. The results suggest that the dynamics of the stochastic VBD depends on a threshold quantity $\clr_0$, the initial size of infectives, and the type of scaling in terms of host population size. The quantity $\clr_0$ for deterministic counterpart of the model, interpreted as threshold condition for infection persistence as is mentioned in the literature for many infectious disease models, can be computed. Different scalings yield different approximations of the model, and in particular, if vectors have much faster dynamics, the effect of the vector dynamics on the host population averages out, which largely reduces the dimension of the model. 
Specific scenarios are also studied using simulations for some fixed sets of parameters to draw conclusions on dynamics. Further stochastic analysis will result in closed formulation of important metrics for disease surveillance such as likelihood of an outbreak and prevalence of a vector borne infectious disease as function of demographical and ecological parameters.\\

\end{abstract}

\section{Introduction}

{\em Vector-Borne Diseases (VBDs)} are infections transmitted by the bite of infected arthropod species, such as mosquitoes, ticks, triatomine bugs, sandflies, and blackflies. It has been shown a century ago that hematophagous (blood-sucking) arthropods transmit particular types of viruses, bacteria, protozoa, and helminths to humans and between animals and humans. Since then, there has been a large number of reports of outbreaks of vector-borne diseases, such as Malaria, Dengue, Chagas diseases, and Leishmaniasis, and the diseases were responsible for more human deaths in the 20th centuries than all other causes combined (cf. \cite{gubler1991insects, philip1973medico}). 
Newly emerging and reemerging vector borne diseases, such as Zika, have recently drawn public attention because of nature of health consequences to new born babies.

Nowadays, changes in land-use, globalization of trade and travel, social upheaval, and intensive new interventions (for example, excessive use of insecticide spraying may change vector behavior and they may become insecticide resistant) have increased the challenges in controlling vector-borne diseases in many regions. VBDs pose serious public health threats throughout the world. According to the World Health Organization (WHO), vector-borne diseases account for more than $17\%$ of all infectious diseases cases, causing more than one million deaths annually (\cite{who2014}). In the USA, west Nile virus, zika, malaria, dengue, chikungunya, eastern equine encephalitis, and St. Louis encephalitis are common diseases that are transmitted by vectors.
The transmission and spread of vector-borne diseases are determined by complex interactions between the hosts (either humans or nonhuman hosts), vectors species (e.g. specific mosquito species), and various pathogens. Important biological properties underlying the transmission of VBDs include survival, development, reproduction of vectors and of pathogens in vectors, vectors' biting rate, and hosts' (humans and nonhumans) behaviors, all of which are associated with environmental conditions and variations (cf. \cite{mubayi2010transmission, pandey2013comparing, towers2016estimate, sheets2010impact, kribs2012role, brauer2016some, gorahava2015optimizing, yong2015agent, malik2012west}). Many of the biological and ecological characteristics of VBDs remain either uncertain or lack enough data to clearly understand its role. 

  {\em Mathematical models} can be used as a tool to systematically help understand the complex behavior of VBD systems in spite of limitation in data related to a VBD. 
The increasing availability of alternate data, from a variety of sources including surveys and entomological field studies, provide the ability to model complex ecosystems enabling human decision-making. Models have the potential to facilitate more accurate assessment of such systems and to provide a basis for more efficient and targeted approaches to treatment and scheduling, through an improved understanding of the disease and transmission dynamics.
Stochastic models 
can be used to incorporate random inherent characteristics of epidemic and provide estimates of variability in model outputs. 
However, the complexity in the models presents many mathematical challenges. The focus of the current study is to provide a unified approach to simplify complex stochastic epidemic models for VBDs, using techniques from probability theory such as the functional law of large numbers (FLLN) and the functional central limit theorem (FCLT). 

There is an extensive literature on stochastic modeling of epidemics. Multivariate Markov jump processes are commonly used in stochastic epidemic models (cf. \cite{bartlett1956deterministic, kurtz1981approximation, ball1983threshold, mubayi2008role}).
Researchers have studied numerous stochastic phenomena such as the distribution of final size of an epidemic (cf. \cite{greenwood2009stochastic}), stochastically sustained oscillations to explain the semi-regular recurrence of outbreaks (cf. \cite{kuske2007sustained}), stochastic amplification of an epidemic (cf. \cite{alonso2007stochastic}), quasi-stationary distributions, which capture variances in endemic states (cf. \cite{allen2008introduction, isham1991assessing, naasell2002stochastic, bani2013influence, VANDOORN20131, Champagnat2016, BRITTON201789}), time to extinction of the disease (cf. \cite{andersson2012stochastic, britton2010stochastic, schwartz2009predicting, mubayi2013contextual}), and critical community sizes needed to have epidemics (cf. \cite{naasell2005new, bartlett1960critical, keeling1997disease, lloyd2005should}). In particular, in both \cite{VANDOORN20131} and \cite{Champagnat2016}, sufficient conditions on the existence and uniqueness of quasi-stationary distributions are studied.

Fluid and diffusion approximations (also known at FLLN and FCLT) are classical approaches in probability to simplify complex stochastic systems. \cite{kurtz1978strong} is a pioneer in developing such approximations for density dependent Markov chains (see also \cite{kurtz1981approximation, Ethier:Kurtz:1986}). Recently, \cite{kang2013} gave a systematic approach for developing FLLNs for multiscale chemical reaction networks. Later, in \cite{kang2014}, the authors provide sufficient conditions for FCLTs of multiscale Markov chains. However, verifying these sufficient conditions for specific complex systems is nontrivial. We show that these sufficient conditions hold in our analysis (for Case II defined below) to establish the FCLT.  In multiscale systems, fluid approximations can achieve dimension reduction via appropriate scaling limit theorems, which can significantly simplify the original complex structure (cf. \cite{kang2013}). It is worth mentioning that dimension reduction methods for VBD models have been studied for deterministic models in literature, although different from the scaling limit theorm approach. For example, \cite{pandey2013comparing} and \cite{Souza2014} used model similar to the VBD model considered here and derived a simple host-only model from a complex vector-host model by assuming that infection dynamics in the vectors are fast compared to those of the hosts.
On the other hand, epidemiological time scales are often used to reduce the dimensionality by identifying components of the model that are evolving naturally at slow, moderate and fast times. These methods are used to corroborate the results of time-scale approximations (see \cite{song2002tuberculosis}).

{\em Procedure in this study:}
To thoroughly explain the methodology, we start with a basic model, the vector-borne SIS model, which has been used to study several vector-borne diseases (cf. \cite{anderson1992infectious}). In the vector-borne SIS model, both host and vector individuals are classified as either {\em susceptible} or {\em infectious}. We assume that the host population size is fixed and is denoted by a positive integer-valued parameter $n$, and the vector population size is a random variable whose mean is $C_0n$, where $C_0$ is a positive integer-valued constant. For each $n$, we have a model and a collection of stochastic processes. We study these models as $n\to\infty$ in two scaling cases. In {\bf Case I}, {hosts and vectors evolve at the same rate} as $n\to\infty$, while in {\bf Case II}, vectors have much faster dynamics than hosts as $n\to\infty.$ For both cases, analogous to FLLN, we develop deterministic processes, which are referred to as {\em fluid limits}, to capture the mean behaviors of the stochastic systems and to study the stability of equilibrium points.
We next, analogous to FCLT, establish {\em diffusion limits}, which are solutions of stochastic differential equations (SDEs), to characterize the statistical fluctuations of the original stochastic systems around their fluid limits. These approximations reduce the dimension of the original system from $3$ to $2$ under Case I, and to $1$ under Case II, which largely simplify the analyses. Especially in Case II, the vector-host system is reduced to a host-only system, where the new transmission parameter is the composite human-to-human transmission rate taking into account both vector-to-host and host-to-vector rates. Comparing the equilibria of the vector-host model and the host-only model can provide understanding of this new transmission parameter in terms of the parameters of the vector-host model.  At last, we apply the fluid and diffusion limits to study the quasi-stationary distribution of the original system. \cite{BRITTON201789} have recently studied a vector-borne SIS model similar to our Case I, but with fixed host and vector population sizes, and they apply similar fluid and diffusion limits to study the quasi-stationary distribution and the extinction times. Model approximations and various sub models are sumarized in Figure \ref{flow_chart}.

\begin{figure}[htbp] 
   \centering
   \includegraphics[width=4in]{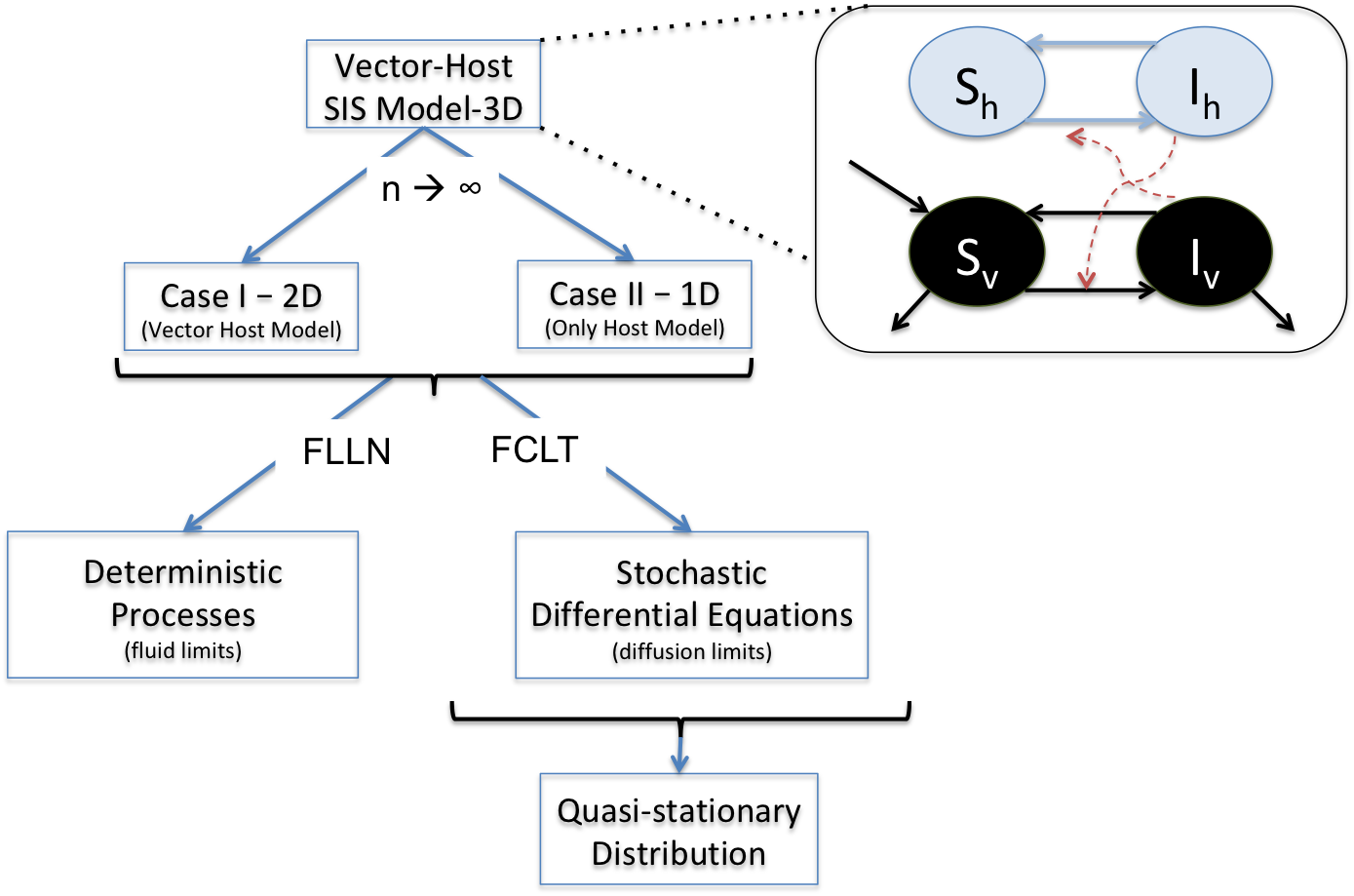} 
   \caption{The original model and its various approximations analyzed in this study}
   \label{flow_chart}
\end{figure}

 {\em The main contribution} of the present paper is as follows. (i) Under the scaling Case II, the stochastic model has fast vector and slow host dynamics. Although there is literature on multiscale deterministic vector-borne epidemic models (see \cite{song2002tuberculosis, pandey2013comparing, Souza2014}), our work is, to our knowledge, the first to study multiscale stochastic vector-borne epidemic models. We rigorously justify the fast and slow scalings, under which the fast vector dynamic is averaged out, and the fluid and diffusion limits for hosts are established. (ii) We provide a convenient and efficient way to study the long time behavior (e.g., quasi-stationary distributions) of the original stochastic system via the fluid and diffusion limits. We draw the following conclusions: When the basic reproduction number $\clr_0$ is greater than $1$: in Case I, since the vector population size has a variance that is linearly growing in time, there is no quasi-stationary distribution, and in Case II, the quasi-stationary distribution is approximately normally distributed. However, we also observe that if the vector population size is fixed, under Case I, the quasi-stationary distribution exists and is approximately normally distributed (also see \cite{BRITTON201789} for similar results).

The rest of the paper is organized as follows. In Section \ref{sec:model}, we introduce the stochastic vector-borne SIS model, explain the two scaling cases, and define the basic reproduction number.  Section \ref{approx} collects all the main results on fluid and diffusion approximations. In Section \ref{sec:quasi}, we study the long time behavior of the stochastic vector-borne SIS model, including quasi-stationary distributions, using the fluid and diffusion limits. We present simulation results in  Section \ref{sec:numerical}. In Section \ref{sec:disc}, a brief discussion is provided. Some fundamental results on equilibrium points of differential equations and matrix exponentials are given in Appendix \ref{app:DE} and \ref{appendix:me}, the stochastic vector-borne SIS model with fixed host and vector population sizes is briefly studied in Appendix \ref{vec-size-fixed}, and all proofs can be found in Appendix \ref{proof}. At last, two notation tables are provided in Appendix \ref{table}.

\section{Methods}\label{sec:model}

\subsection{Model description}
Mathematical models have a long history of providing important insight into disease dynamics and control.
We study the vector-borne SIS model, in which both host and vector individuals are classified as either {\em susceptible} or {\em infectious}. We assume that the infected hosts can recover and become susceptible immediately again, while vectors stay infectious till they die. The host population size is assumed to be fixed and is denoted by the positive integer-valued parameter $n$, and the initial vector population size is equal to $C_0n$, where $C_0$ is a positive constant. Denote by $S_h^n(t)$ and $I_h^n(t)$ the numbers of susceptible and infected hosts at time $t$, and similarly, $S_v^n(t)$ and $I_v^n(t)$ the numbers of susceptible and infected vectors at time $t$.

We assume that the vector biting rate, birth rate, and death rate are all constant, and there is no feeding preference or host competence. Let $\beta^n_h$ denote the disease transmission rate to hosts from a typical infecious vector. Noting that $I^n_v(t)/n$ gives the average number of vectors per host,  we see that $\beta^n_hI^n_v(t)/n$ measures the infection rate for susceptible hosts. Next denote by $\gamma^n_h$ the recovery rate for infected hosts. For vectors, let $\gamma^n_v$ represent the equal birth and death rate per vector, and $\beta^n_v$ be the disease transmission rate to vectors from a typical infectious host. The ratio $I^n_h(t)/n$ can be interpreted as the probability that a vector contacts an infectious host, and so $\beta^n_vI^n_h(t)/n$ measures the infection rate for susceptible vectors.

We model the infection, recovery, birth, and death processes using independent unit-rate Poisson processes (see Theorem 4.1 of Chapter 6 in \cite{Ethier:Kurtz:1986}). More precisely, we have the following system of equations: For $t\ge 0,$
\begin{align}
S_h^n(t) & = S_h^n(0)  - N_1^n\left(\beta_h^n\int_0^t \frac{I_v^n(u) S_h^n(u)}{n} \D u \right)  + N_2^n\left(\gamma_h^n \int_0^t  I_h^n(u) \D u \right), \label{sys-1}\\
I_h^n(t) & = I_h^n(0) + N_1^n\left(\beta_h^n \int_0^t \frac{I_v^n(u) S_h^n(u)}{n} \D u \right)  - N_2^n\left(\gamma_h^n \int_0^t  I_h^n(u) \D u \right), \label{sys-2}
\end{align}
and
\begin{align}
\begin{split}
S_v^n(t) & = S_v^n(0) + N_3^n\left(\gamma_v^n \int_0^t [S^n_v(u) + I^n_v(u)] \D u \right) - N_4^n\left(\beta_v^n  \int_0^{ t} \frac{I_h^n(u) S_v^n(u)}{n} \D u \right) \\
& \quad - N_5^n \left(\gamma_v^n \int_0^t  S_v^n(u) \D u \right),
\end{split} \label{sys-3}\\
I_v^n(t) & = I_v^n(0) + N_4^n\left(\beta_v^n  \int_0^{ t} \frac{I_h^n(u) S_v^n(u)}{n} \D u \right)- N_6^n \left(\gamma_v^n \int_0^t  I_v^n(u) \D u \right), \label{sys-4}
\end{align}
where $N^n_i, i=1, 2, \ldots, 6$, are independent unit-rate Poisson processes, which are independent of the initials $S^n_h(0), I^n_h(0), S^n_v(0)$, and $I^n_v(0)$.

We note that the host population size $S^n_h(t)+I^n_h(t)$ is equal to $n$ for all $t\ge 0$, and the vector population size $S^n_v(t)+I^n_v(t)$ is a random variable. However, the expected vector population size $\EE(S^n_v(t)+I^n_v(t))$ is equal to its initial value $C_0n$ for all $t\ge 0$. The epidemic system can be completely described by a $3$-dimensional process $(I^n_h, S^n_v, I^n_v).$ From the formulation in \eqref{sys-1} -- \eqref{sys-4}, it can be seen that $(I^n_h, S^n_v, I^n_v)$ is a continuous-time Markov chain (CTMC) with infinite state space $\{0, 1, \ldots, n\}\times \mathbb{N}_0 \times \mathbb{N}_0$, where $ \mathbb{N}_0$ is the set of non-negative integers. We also observe that once the infection process $(I^n_h, I_v^n)$ reaches the state $(0, 0)$, it will stay at $(0,0)$ forever, and the process $S^n_v$ will then become a linear birth and death process with equal birth and death rate $\gamma^n_v$ per individual. Some long time properties of $(I^n_h, S^n_v, I^n_v)$ will be studied in Section \ref{sec:quasi}.

\subsection{Asymptotic scales}\label{sec:scaling}

We consider two asymptotic scales as the host population size becomes large, i.e., $n\to\infty$. In Case I, vectors and hosts evolve on the same scale, while in Case II, vectors evolve much faster than hosts. More precisely, let $\beta_h, \gamma_h, \beta_v, \gamma_v$ be positive constants. We have the following two cases.
\begin{itemize}
\item[] Case I: {\em Both hosts and vectors evolve at rate $\mathcal{O}(1)$ as $n\to\infty$.}
\begin{align*}
& \mbox{Host:} \ \beta_h^n \to \beta_h, \ \gamma_h^n \to \gamma_h,\\
& \mbox{Vector:} \ \beta_v^n \to \beta_v, \ \gamma_v^n \to \gamma_v.
\end{align*}

\item[] Case II: {\em  Vectors evolve much faster at rate $\mathcal{O}(\alpha(n))$ and hosts evolve at rate $\mathcal{O}(1)$ as $n\to\infty$.}
\begin{align*}
& \mbox{Host:} \ \beta_h^n \to \beta_h, \ \gamma_h^n \to \gamma_h,\\
& \mbox{Vector:} \ \frac{\beta_v^n}{\alpha(n)} \to \beta_v, \ \frac{\gamma_v^n}{\alpha(n)} \to \gamma_v,
\end{align*}
 where $0 < \alpha(n)/\sqrt{n} \to \infty$ and $\alpha(n)/n \to 0$ as $n\to\infty$.
\end{itemize}

To understand the above scaling cases, let's assume the time unit is one day. During one day, we measure the transmission rate $\beta^n_h$ and recovery rate $\gamma^n_h$ for hosts, and transmission rate $\beta^n_v$ and equal birth and death rate $\gamma^n_v$ for vectors. Due to the interaction between hosts and vectors, as the host population size $n$ varies, these parameters also vary, and so we let the parameters depend on $n$.
The scaling Case I simply says as the host population size grows, these parameters approach some steady values. Under Case II, we observe that the transmission rate $\beta^n_v$ and the birth and death rate $\gamma^n_v$ are much larger than the parameters for hosts, as $n\to\infty$. To understand the scaling parameter $\alpha(n)$, one could sample a sequence of parameter estimates $\{(\beta^n_{h}, \gamma^n_{h}, \beta^n_{v}, \gamma^n_{v})\}_{n=n_0}^N$ for $N-n_0$ different population sizes, and plot the ratio $\gamma^n_{v}/\gamma^n_{h}$ as a function of $n$ to observe the order of $\alpha(n)$, because $\gamma^n_{v}/\gamma^n_{h} \approx \alpha(n) \gamma_h/\gamma_v$. {Mathematically, we require $\alpha(n)/\sqrt{n} \to \infty$ and $\alpha(n)/n \to 0$ as $n\to\infty$, e.g., $\alpha(n)=n^{2/3}$, to develop the scaling limith theorems.}

\subsection{Basic reproduction number}\label{sec:reprod}

The basic reproduction number, $\clr_0$, is defined as the expected number of secondary cases caused by one infectious individual introduced into a susceptible population during his/her infectious period. It is a measure of the success of an invasion into a population; if $\clr_0>1$, a larger outbreak and endemic is possible, whereas if $\clr_0<1$, the infection will certainly die out in the long run. The reproduction number of the VBD is defined in a similar way, and could depend on vector mortality rate, pathogen development rate, and host competence and recovery rate (cf. \cite{lord1996vector}). 

Using next generation matrix approach (\cite{van2002reproduction}), it is straight forward to compute the basic reproduction number, $R_0$ for the deterministic counterpart of the Model \eqref{sys-1} -- \eqref{sys-4}. The $\clr_0$ is given as
$$\mathcal{R}_0 = \sqrt{ \left ( \frac{\beta_h}{\gamma_v} \right ) \cdot \left ( \frac{C_0n}{n} \frac{\beta_v}{\gamma_h} \right )} .$$
We note that $\beta_h/\gamma_v$ represents the average number of newly infected host individuals produced by a typical infectious vector during its mean survival time period, and $C_0\beta_v/\gamma_h$ represents the average number of newly infected vector individuals produced by a typical infectious host during its mean infection time period. Thus, the basic reproduction number $\clr_0$ (geometric mean) is the average number of newly infected host individuals generated by a typical infectious host individual via vector-host and host-vector transmissions.

\section{Model Simplifications: Fluid and Diffusion Approximations}\label{approx}

The transient behavior of $(I^n_h, S^n_v, I^n_v)$ is rather complex and cannot be analyzed easily. In this section, we simplify the orginal epidemic model by establishing the fluid and diffusion approximations through scaling limit theorems for the system equations \eqref{sys-1} -- \eqref{sys-4}. We first consider the fluid scaling, under which the processes are divided by the host population size $n$. These rescaled processes are referred to as {\em fluid scaled processes}. Using FLLN methods, we establish the deterministic limits of the fluid scaled processes in Theorems \ref{lln-1} and \ref{lln-2}. These limits are called {\em fluid limits}, and they capture the average behavior of the fluid scaled processes as the host population size grows to infinity. To characterize the statistical fluctuations of the fluid scaled processes around their fluid limits, we next study the difference between the fluid scaled processes and their fluid limits, which we refer to as the deviation process. We show that the suitably scaled deviation processes converge weakly to SDEs (see Theorems \ref{fclt-1} and \ref{fclt-2}), which will be referred to as {\em diffusion limits}.

\subsection{Fluid approximations}\label{fluid-approx}

We define the following fluid scaled processes. For $t\ge 0,$
\begin{align}\label{fluid-scaling}
 \bar S_h^n(t) = \frac{S_h^n(t)}{n},  \ \bar I_h^n(t) = \frac{I_h^n(t)}{n}, \ \bar S_v^n(t) = \frac{S_v^n( t)}{n}, \ \bar I_v^n(t) = \frac{I_v^n( t)}{n}.
\end{align}
In particular, for $t\ge 0$, the quantities $\bar S_h^n(t)$ and $\bar I_h^n(t)$ represent the densities of susceptible and infected host individuals at time $t$, respectively, and $\bars^n_v(t)$ and $\bar I_v^n(t)$ represent the numbers of susceptible and infected vectors per host at time $t$, respectively.
We also observe that
\begin{align}\label{relation-1}
\bar I_h^n(t) + \bar S_h^n(t) =1, \ \ t\ge 0.
\end{align}
Thus, under the fluid scaling, the system can be completely described by $(\bar I_h^n, \bar S^n_v, \bar I^n_v)$.

In the following, we present the fluid limits in both scaling cases.  The stability properties of the fluid limits are also studied. (The definitions of different stability concepts are provided in Appendix \ref{app:DE}.) In Case I, as the host population size $n\to\infty$, the fluid scaled vector population size $\bar S^n_v(t) + \bar I^n_v(t)$ approaches the constant $C_0$, and {\em the system dimension is reduced to two}.

\begin{theorem}\label{lln-1} In Case I, assume that $\EE[( \bar I_h^n(0), \bar I_v^n(0)) - (i_h(0), i_v(0))]\to 0$, as $n\to\infty$, for some constant vector $ (i_h(0), i_v(0)) \in \{(x,y): 0\le x \le 1, 0\le y\le C_0\}$. Then for any $t\ge 0$, \[\EE\left[\sup_{0\le s \le t}|(\bar I_h^n(s), \bar S_v^n(s), \bar I_v^n(s)) - (i_h(s), s_v(s), i_v(s))|\right]\to 0, \ \mbox{as $n\to\infty$,}\]
where $s_v(t) = C_0 - i_v(t)$,  and $(i_h, i_v)$ is the unique solution of the following ODEs with initial value $(i_h(0), i_v(0))$. For $t\ge 0$,
\begin{align}
\frac{\D i_h(t)}{\D t} & = \beta_h i_v(t) (1- i_h(t)) - \gamma_h i_h(t), \label{ode1-1}\\
\frac{\D i_v(t)}{\D t} & = \beta_v i_h(t) (C_0 - i_v(t)) - \gamma_v i_v(t).\label{ode2-1}
\end{align}

\end{theorem}

\begin{theorem}\label{lln-stability-I}
The ODEs in \eqref{ode1-1} and \eqref{ode2-1} have two equilibrium points
\[
E^f =(0, 0) \ \ \mbox{and} \ \ E^e = \left(\frac{C_0\beta_h\beta_v - \gamma_h\gamma_v}{C_0\beta_h\beta_v + \beta_v\gamma_h}, \ \ \frac{C_0\beta_h\beta_v - \gamma_h\gamma_v}{\beta_h\beta_v + \beta_h\gamma_v}\right),
\]
and
\begin{itemize}
\item[\rm (i)] the disease free equilibrium $E^f$ is globally asymptotically stable when $\clr_0 \le 1$, it is globally exponenitally stable when $\clr_0 <1$,  and it is unstable when $\clr_0 >1$;
\item[\rm (ii)] when $\mathcal{R}_0 >1$, the endemic equilibrium $E^e$ is locally asymptotically stable, and it is globally asymptotically stable when $(i_h(0), i_v(0)) \in (0, 1] \times (0, C_0]$.
\end{itemize}
\end{theorem}

In Case II, the vectors evolve much faster than the hosts as $n\to\infty.$ In fact, as $n\to\infty$, at each time point $t$, the vectors stay in their equilibrium state, which depends on the state of the hosts at $t$. Accordingly, the system state can be determined by the state of the hosts, and {\em the system dimension is reduced to one}. To characterize the equlibrium state of $(\bars^n_v, \bari^n_v)$, we define for $t\ge 0$ and a Borel set $E\subset \RR_+^2$, a measure $\Gamma^n$ on $[0,t]\times E$ as follows:
\[
\Gamma^n([0,t]\times E) = \int_0^t 1_{\{(\bars^n_v(s),\bar I_v^n(s)) \in E\}} \D s.
\]

\begin{theorem}\label{lln-2} Under Case II, assume that $\EE[|(\bar I_h^n(0), \bari^n_v(0)) - (i_h(0), i_v(0))|] \to 0$, as $n\to\infty$, for some constant vector $(i_h(0), i_v(0))\in \{(x,y,z): 0\le x \le 1, 0\le y\le C_0\}$. Then for any $t\ge 0$,
\[ \EE\left[\sup_{0\le s \le t}|\bar I_h^n(s) - i_h(s)|\right] \to 0, \ \mbox{as $n\to\infty$,}\] where $i_h$ is the unique solution to the following ODE with initial value $i_h(0)$. For $t\ge 0,$
\begin{align}
\frac{\D i_h(t)}{\D t} & = \frac{C_0\beta_h\beta_v i_h(t)  (1- i_h(t))}{\beta_v i_h(t) + \gamma_v} - \gamma_h i_h(t). \label{ode-2}
\end{align}
Furthermore, for $E\in \mathcal{B}(\RR_+^2)$,
\begin{align}\label{fast-vectors}
\Gamma^n([0,t]\times E) \to \int_0^t 1_{\left\{\left(\frac{C_0\gamma_v}{\beta_vi_h(t)+\gamma_v}, \ \frac{C_0\beta_vi_h(t)}{\beta_vi_h(t)+\gamma_v}\right)\in E\right\}} \D s.
\end{align}
\end{theorem}

\begin{theorem}\label{lln-stability-II}
The ODE in \eqref{ode-2} has two equilibrium points
\[
E^f = 0 \ \ \mbox{and} \ \ E^e = \frac{C_0\beta_h\beta_v -\gamma_h\gamma_v}{C_0\beta_h\beta_v + \beta_v\gamma_h},
\]
and
\begin{itemize}
\item[\rm (i)] the disease free equilibrium $E^f$ is globally asymptotically stable when $\clr_0\le 1$, it is globally exponentially stable when $\clr_0 < 1$, and it is unstable when $\clr_0 >1$;
\item[\rm (ii)] when $\mathcal{R}_0 >1$, the endemic equilibrium $E^e$ is locally asymptotically stable, and it is globally asymptotically stable when $i_h(0)\in (0,1]$.
\end{itemize}

\end{theorem}

\begin{remark}\hfill
\begin{itemize}
\item[\rm (i)] In scaling Case I, the hosts and vectors evolve at the same rate as $n\to\infty$ towards their fluid limits defined in \eqref{ode1-1} and \eqref{ode2-1}. The equilibrium points of these ODEs are derived by letting the derivative on the LHS of \eqref{ode1-1} and \eqref{ode2-1} be $0$, and solve $\beta_h i_v^* (1- i_h^*) - \gamma_h i_h^* =0$ and $\beta_v i_h^*(C_0 - i_v^*) - \gamma_v i_v^* =0$ for $(i^*_h, i^*_v).$
\item[\rm (ii)]  {In scaling Case II, the vectors have much shorter life cycles. From \eqref{fast-vectors}, in the fluid limit at time $t>0$, the vectors are in the quasi-equilibrium state $(\frac{C_0\gamma_v}{\beta_vi_h(t)+\gamma_v}, \ \frac{C_0\beta_vi_h(t)}{\beta_vi_h(t)+\gamma_v})$, which depends on the state of the hosts $i_h(t).$}
\end{itemize}
\end{remark}

\subsection{Diffusion approximations} \label{diff-approx}

In this section, we characterize the statistical fluctuations of the fluid scaled processes around their fluid limits that are developed in Section \ref{fluid-approx}. To achieve this, we define the diffusion scaled processes: For $t\ge 0$,
\begin{equation}\label{diff-scaling}
\begin{aligned}
\hat S^n_h(t) & = \sqrt{n} (\bar S^n_h(t) - s_h(t)), \\
\hat I^n_h(t) & = \sqrt{n} (\bar I^n_h(t) - i_h(t)), \\
\hat S^n_v(t) & = \sqrt{n} (\bar S^n_v(t) - s_v(t)), \\
\hat I^n_v(t) & = \sqrt{n} (\bar I^n_v(t) - i_v(t)).
\end{aligned}
\end{equation}
From \eqref{relation-1}, and the fact that $i_h(t)+s_h(t)=1$, we note that $\hat I^n_h(t) + \hat S^n_h(t) = 0$, and so it suffices to study $(\hat I^n_h, \hat S^n_v, \hat I^n_v).$ Further noting that $S^n_v(0)+ I^n_v(0) = C_0 n$ and $s_v(t) + i_v(t) = C_0$ for $t\ge 0$, so we have $\hats^n_v(0) + \hati^n_v(0) =0.$

\begin{theorem}\label{fclt-1} Consider Case I, and assume that $(\hat I^n_h(0), \hat S^n_v(0), \hat I^n_v(0))$ {converges in distribution to some random variable} $(\hat I_h(0), \hat S_v(0), \hat I_v(0))$ with $\hats_v(0) + \hati_v(0)=0$, and that
\begin{align}\label{fclt-cond-1}
\sqrt{n}(\beta^n_h - \beta_h) \to \hat\beta_h, \ \ \sqrt{n}(\gamma^n_h - \gamma_h) \to \hat\gamma_h, \ \ \sqrt{n}(\beta^n_v - \beta_v) \to \hat\beta_v, \ \ \sqrt{n}(\gamma^n_v - \gamma_v) \to \hat\gamma_v.
\end{align}
 Then $(\hat I_h^n, \hat S^n_v, \hat I_v^n)$ converges in distribution to $(\hat I_h, \hat S_v, \hat I_v)$, where $(\hati_h, \hat S_v, \hati_v)$ is the unique solution to the following SDEs: For $t\ge 0,$
\begin{align}
\hat I_h(t) & = \hat I_h(0) + \int_0^t [\beta_hs_h(u)\hat I_v(u) - (\beta_h i_v(u)+\gamma_h)\hat I_h(u)]\D u+ \int_0^t [\hat\beta_h i_v(u)s_h(u)-\hat\gamma_hi_h(u)]\D u, \label{host-limit-1-1}\\
& \quad  + \int_0^t \sqrt{\beta_h i_v(u)s_h(u) + \gamma_h i_h(u)} \D B_1(u), \label{host-limit-1-2}\\
\hat S_v(t) & = \hat S_v(0) + \int_0^t [-\beta_v s_v(u)\hat I_h(u) - \beta_v i_h(u)\hat S_v(u) + \gamma_v \hat I_v(u)]\D u+\int_0^t [-\hat\beta_v i_h(u)s_v(u)+\hat\gamma_vi_v(u)]\D u, \nonumber \\
& \quad  + \int_0^t \sqrt{\gamma_v(C_0 + s_v(u))} \D B_2(u)- \int_0^t  \sqrt{\beta_vi_h(u)s_v(u)} \D B_3(u),\nonumber \\
\hat I_v(t) & = \hat I_v(0) + \int_0^t [\beta_v s_v(u)\hat I_h(u) + \beta_v i_h(u)\hat S_v(u) - \gamma_v \hat I_v(u)]\D u +  \int_0^t [\hat\beta_v i_h(u)s_v(u)-\hat\gamma_vi_v(u)]\D u\nonumber \\
& \quad  +   \int_0^t \sqrt{\beta_v i_h(u)s_v(u)} \D B_3(u) +\int_0^t  \sqrt{\gamma_v i_v(u)} \D B_4(u),\nonumber
 \end{align}
 with $B_i, i=1,2,3,4,$ being four independent standard Brownian motions, which are independent of $(\hati_h(0), \hats_v(0), \hati_v(0))$.
\end{theorem}

\begin{theorem}\label{fclt-2}
Consider Case II, and assume that $\hat I^n_h(0)$ converges in distribution to some random variable $\hat I_h(0)$, and that
\begin{align}\label{fclt-cond-2}
\sqrt{n}(\beta^n_h - \beta_h) \to \hat\beta_h, \ \sqrt{n}(\gamma^n_h - \gamma_h) \to \hat\gamma_h, \ \sqrt{n}\left(\frac{\beta^n_v}{\alpha(n)} - \beta_v\right) \to \hat\beta_v, \ \sqrt{n}\left(\frac{\gamma^n_v}{\alpha(n)} - \gamma_v\right) \to \hat\gamma_v.
\end{align}
Then $\hat I_h^n$ converges weakly to $\hat I_h$, where $\hat I_h$ is the unique solution of the following SDE: For $t\ge 0,$
\begin{align}
{\hati_h(t) = \hati_h(0) + \int_0^t D(i_h(u), \hati_h(u)) \D u + \int_0^t \sqrt{ \beta_hi^*_v(u)(1- i_h(u)) + \gamma_h i_h(u)} \ \D B(u), \label{host-limit-2-1} }
\end{align}
with $B$ being a standard Brownian motion independent of $\hati_h(0)$, $i^*_v(t) = \frac{C_0\beta_v i_h(t)}{\beta_v i_h(t)+\gamma_v}$, and
\begin{equation}\label{host-limit-2-2}
\begin{aligned}
& D(i_h(t), \hati_h(t)) \\
& =\hat\beta_h i_v^*(t)(1-i_h(t)) - \hat\gamma_h i_h(t) \\
& \quad + \frac{\beta_h(1- i_h(t))(C_0 \hat\beta_v i_h(t) - \hat\beta_vi_h(t)i^*_v(t) - \hat\gamma_v i^*_v(t))}{\beta_v i_h(t)+\gamma_v}  \\
& \quad +\frac{((C_0\beta_h\beta_v - \gamma_h\gamma_v)\gamma_v - (C_0\beta_h\beta_v + \gamma_h\beta_v)\beta_v i_h(t)^2-2(C_0\beta_h\beta_v + \gamma_h\beta_v)\gamma_v i_h(t))\hati_h(t)}{(\beta_v i_h(t)+\gamma_v)^2}.
\end{aligned}
\end{equation}

\end{theorem}

\begin{remark}\label{rem-diff}\hfill
\begin{itemize}
\item[\rm (i)] The assumptions in \eqref{fclt-cond-1} and \eqref{fclt-cond-2} say that the parameters associated with the host population size, $n$, should converge to the steady values at the same rate as or faster than $O(1/\sqrt{n})$.
\item[\rm (ii)] {Comparing the diffusion limits $\hati_h$ under the two different scaling cases in Theorems \ref{fclt-1} and \ref{fclt-2}, we observe that the diffusion coefficients (i.e., the coefficients before the Brownian motions) in \eqref{host-limit-1-2} and \eqref{host-limit-2-1} are similar except that the equilibrium $i^*_v$ appears in \eqref{host-limit-2-1} and the regular fluid limit $i_v$ in \eqref{host-limit-1-2}. The diffusion coefficient is in this sense more complex in Case I than in Case II.  However, compared to \eqref{host-limit-1-1} (Case I), the drift of $\hati_h$ in \eqref{host-limit-2-2} (Case II) is more complex because it contains all the quasi-equilibrium information of $\hati_v$ and $\hats_v$. }

\item[\rm (iii)] In Theorem \ref{fclt-1}, we observe that
\[
\hats_v(t) + \hati_v(t) = \int_0^t \sqrt{\gamma_v(C_0 + s_v(u))} \D B_2(u) + \int_0^t  \sqrt{\gamma_v i_v(u)} \D B_4(u),
\]
which is normally distributed with mean $0$ and variance $\int_0^t[\gamma_v(C_0 + s_v(u)) +\gamma_v i_v(u)] \D u= 2\gamma_vC_0t.$ From \eqref{dist-approx}, below, we see that $S^n_v(t)+I^n_v(t)\approx n (s_v(t)+i_v(t)) + \sqrt{n}(\hats_v(t)+\hati_v(t)) = C_0 n + \sqrt{n}(\hats_v(t)+\hati_v(t)).$ Thus $S^n_v(t)+I^n_v(t)$ is approximately normally distributed with mean $C_0n$ and variance $2n\gamma_vC_0t.$
\end{itemize}
\end{remark}

\section{Quasi-stationary distributions}\label{sec:quasi}

There are many stochastic systems arising in epidemic modeling in which the disease eventually {\em dies out}, yet appears to be persistent over any reasonable time scale. We are often interested in such long time behavior of a stochastic epidemic process which has zero as an absorbing state for the infected population, almost surely. The hitting time of this state, namely the {\em extinction time}, can be large compared to the physical time and the population size can fluctuate for a large amount of time before extinction actually occurs. This phenomenon can be understood by the study of {\em quasi-stationary distributions}. The quasi-stationary distribution or conditional limiting distribution has proved to be a powerful tool for describing properties of {Markov population processes such as recurrent epidemics as modeled} in \cite{darroch1965quasi, kryscio1989extinction}. The computation of these distributions is critical, as the expected time to extinction starting from quasi-stationarity and the {\em critical community size} for epidemic to die out within a specified time for various ranges of $\clr_0$ can also be then computed (cf. \cite{VANDOORN20131, mubayi2013contextual}). In finite-state systems the existence of a quasi-stationary distributions is guaranteed. However, in infinite-state systems this may not always be so (cf. \cite{pollett1995determination, VANDOORN20131, Champagnat2016}).

In this section, novel results on long time properties of $(I^n_h, S^n_v, I^n_v)$, including quasi-stationarity, are obtained for large $n$ based on the fluid and diffusion approximations developed in Section \ref{approx}.
For simplicity, we assume that $\hat\beta_h=\hat\gamma_h=\hat\beta_v=\hat\gamma_v =0$ in \eqref{fclt-cond-1} and \eqref{fclt-cond-2}, which happens when the parameters associated with host population $n$ converges to their limits faster than $O(1/\sqrt{n}).$

\subsection{Case II}

Let's first study Case II, in which the system dimension is reduced to one, and it suffices to consider the process $I^n_h$. From the diffusion scaling in \eqref{diff-scaling}, we have
\begin{align*}
I^n_h(t) = ni_h(t) + \sqrt{n}\hati^n_h(t).
\end{align*}
Using the diffusion approximation developed in Theorem \ref{fclt-2}, we have  for large $n$,
\begin{align}
I^n_h(t) \approx  \underbrace{ni_h(t)}_{\mbox{Fluid approximation}} +  \underbrace{\sqrt{n}\hati_h(t),}_{\mbox{Diffusion approximation}} \  \mbox{in distribution.}
\end{align}
The above approximation implies that the long time behavior of $I^n_h$ can be studied through the analysis of $i_h$ and $\hati_h$. The stability of the equilibrium points of $i_h$ is provided in Theorem \ref{lln-stability-II}. In the following we study the distribution of $\hati_h(t)$ for large enough $t$ such that $i_h(t)$ is near its stable equilibrium point.

Define for $t\ge 0$,
\begin{align*}
C(t) & = \frac{(C_0\beta_h\beta_v - \gamma_h\gamma_v)\gamma_v - (C_0\beta_h\beta_v + \gamma_h\beta_v)\beta_v i_h(t)^2-2(C_0\beta_h\beta_v + \gamma_h\beta_v)\gamma_v i_h(t)}{(\beta_v i_h(t)+\gamma_v)^2}, \\
\sigma(t) & = \sqrt{ \beta_hi^*_v(u)(1- i_h(u)) + \gamma_h i_h(u)},
\end{align*}
where $i^*_v(t) = \frac{C_0\beta_v i_h(t)}{\beta_v i_h(t)+\gamma_v}$. {Then $\hat I_h$ satisfies the following $1$-dimensional linear SDE. For $t\ge 0,$}
\begin{align}\label{SDE-CASEII}
\D \hat I_h(t) =  C(t) \hat I_h(t) \D t + \sigma(t) \D B(t).
\end{align}

When $\mathcal{R}_0\le 1$, the disease free equilibrium point $ 0$ is globally asymptotically stable for the fluid limit $i_h(t)$. Thus there exists $t^*>0$ such that $i_h(t)\approx  0$ for $t\ge t^*>0.$  So when $t\ge t^*$, we have
\begin{align*}
C(t) \approx C_f = \frac{C_0\beta_h\beta_v - \gamma_h\gamma_v}{\gamma_v} \le 0, \ \ {\sigma(t) \approx \sigma_f = 0,}
\end{align*}
and the SDE in \eqref{SDE-CASEII} is reduced to be the following ODE
\[
\D \hati_h(t) \approx C_f \hati_h(t) \D t.
\]
Solving the above ODE, we have
\begin{align}\label{case-2-R<1}
\hat I_h(t) \approx e^{C_f (t-t^*)} \hat I_h(t^*).
\end{align}
Thus when $C_f<0$, i.e. $\clr_0<1$, $\hat I_h(t)$ approaches $0$ exponentially fast.

When $\mathcal{R}_0 > 1$, and $i_h(0) = E^e$, where $E^e$ is the endemic equilibrium point of the fluid limit $i_h(t)$. Noting that $E^e$ is globally asymptotically stable, it follows that $i_h(t) = E^e$ for $t \ge 0.$ Then for $t\ge 0,$
\begin{align}\label{eep}
C(t) \equiv C_e = \frac{-(C_0\beta_h\beta_v - \gamma_h\gamma_v)}{\beta_v i^*_h + \gamma_v} <0, \ \ \sigma(t) \equiv \sigma_e = \sqrt{2\gamma_h i^*_h},
\end{align}
and the SDE in \eqref{SDE-CASEII} is reduced to be the following homogeneous linear SDE:
\[
\D \hati_h(t) = C_e \hati_h(t) \D t + \sigma_e \D B(t),
\]
We also note that $\hati_h$ is a one-dimensional Ornstein-Uhlenbeck process. 
Solving the above SDE, we have
\begin{align*}
\hat I_h(t) = e^{C_et} \hati_h(0) +  \int_{0}^t e^{C_e(t-s)} \sigma_e\D B(s).
\end{align*}
It follows that 
\begin{align*}
\var(\hati_h(t)) = \var(e^{C_et} \hati_h(0)) + \var\left( \int_{0}^t e^{C_e(t-s)} \sigma_e\D B(s)\right) = e^{2C_e t} \var(\hati_h(0)) + \int_0^t e^{2C_e(t-s)} \D s,
\end{align*}
and 
{the limiting covariance of $\hati_h$ is then given by}
\begin{align*}
\Sigma^*& \equiv \lim_{t\to\infty} \var(\hati_h(t))   = \lim_{t\to\infty}\int_{0}^t e^{2C_e(t-s)}\sigma_e^2 \D s  = - \frac{\sigma_e^2}{2C_e}.
\end{align*}

\begin{conjecture}
For large host population size $n$,
\begin{itemize}
\item[\rm (i)] when $\clr_0 < 1$, it follows from \eqref{case-2-R<1} and Theorem \ref{lln-stability-II} (i) that $I^n_h(t)$ approaches $0$ exponentially fast as $t\to\infty$;
\item[\rm (ii)] when $\clr_0 > 1$, the quasi-stationary distribution of $I^n_h$ can be approximated by a normal distribution with mean $n E^e$ and covarance matrix $n\Sigma^*$, where $E^e=\frac{C_0\beta_h\beta_v -\gamma_h\gamma_v}{C_0\beta_h\beta_v + \beta_v\gamma_h}$ is the endemic equilibrium point, and $\Sigma^*=- \frac{\sigma_e^2}{2C_e}$ with $C_e$ and $\sigma_e$ defined in \eqref{eep}.
\end{itemize}
\end{conjecture}

\subsection{Case I}\label{quasi-case-1}

Case I is more complex than Case II since we need to study the $3$-dimensional diffusion limit $(\hati_h, \hats_v, \hati_v)$. In particular, some results on matrix exponentials, which are provided in Appendix \ref{appendix:me},  are required to solve linear ODEs and SDEs.

From the diffusion scaling in \eqref{diff-scaling}, we have
\begin{align*}
(I^n_h(t), S_v^n(t), I_v^n(t)) = n(i_h(t), s_v(t), i_v(t)) + \sqrt{n}(\hati^n_h(t), \hats^n_v(t), \hati_v^n(t)).
\end{align*}
Using the diffusion approximation developed in Theorem \ref{fclt-1}, we have  for large $n$,
\begin{align}\label{dist-approx}
(I^n_h(t), S^n_v(t), I_v^n(t)) \approx  \underbrace{n(i_h(t), s_v(t), i_v(t))}_{\mbox{Fluid approximation}} +  \underbrace{\sqrt{n}(\hati_h(t), \hats_v(t), \hati_v(t)),}_{\mbox{Diffusion approximation}} \  \mbox{in distribution.}
\end{align}
The stability of the equilibrium points of $(i_h, s_v, i_v)$ is provided in Theorem \ref{lln-stability-I}. In the following we focus on the limiting distribution of $(\hati_h(t), \hats_v(t), \hati_v(t))$.

For $t\ge 0,$ let
\begin{align*}
C(t) & = \begin{pmatrix} - (\gamma_h +\beta_h i_v(t)) & 0 & \beta_h s_h(t) \\ -\beta_vs_v(t) & -\beta_vi_h(t)& \gamma_v\\ \beta_vs_v(t) & \beta_v i_h(t) & - \gamma_v   \end{pmatrix}, \\
\sigma(t) & = \begin{pmatrix} \sqrt{\beta_h i_v(t) s_h(t)+\gamma_h i_h(t)}  & 0 & 0 & 0 \\ 0 & \sqrt{\gamma_v(C_0+s_v(t))} & -\sqrt{\beta_vi_h(t)s_v(t)} & 0 \\ 0 & 0 & \sqrt{\beta_v i_h(t)s_v(t)} & \sqrt{\gamma_v i_v(t)}\end{pmatrix},\\
\B(t) & = \begin{pmatrix} B_1(t), B_2(t), B_3(t), B_4(t)\end{pmatrix}',
\end{align*}
and
\[
\X(t) = (\hati_h(t), \hats_v(t), \hati_v(t)).
\]
Then the diffusion limit process $\X$ satisfies the following $3$-dimensional linear SDE.
\begin{align}\label{diff-sde}
\D \X(t) & = C(t)\X(t) \D t + \sigma(t) \D \B(t).
\end{align}

 If $\clr_0 \le 1$, then the disease free equilibrium point $ (0,C_0, 0)$ is globally asymptotically stable for the fluid limit $(i_h(t), s_v(t), i_v(t))$. Thus there exists $t^*>0$ such that $(i_u(t), s_v(t), i_v(t))\approx  (0, C_0, 0)$ for $t\ge t^*>0.$ Then for $t\ge t^*,$
\begin{align*}
C(t) \approx C_f\equiv \begin{pmatrix} -\gamma_h & 0 & \beta_h  \\ -C_0\beta_v & 0 & \gamma_v  \\C_0\beta_v & 0 & -\gamma_v   \end{pmatrix}, \ \ \sigma(t) \approx \sigma_f\equiv \begin{pmatrix} 0 & 0 & 0 & 0 \\ 0 & \sqrt{2C_0\gamma_v} & 0 & 0 \\ 0 & 0& 0 & 0 \end{pmatrix}.
\end{align*}
Thus, when $t\ge t^*$, $(\hati_h(t), \hati_v(t))$ can be approximated by the following $2$-dimensional homogeneous linear ODE:
\begin{align*}
\D \hati_h(t) & \approx [-\gamma_h \hati_h(t)+\beta_h\hati_v(t)] \D t, \\
\D \hati_v(t) & \approx [C_0\beta_v \hati_h(t) - \gamma_v \hati_v(t)] \D t.
\end{align*}
Solving the ODE, we have the following approximation:
\begin{align}\label{approx-I}
(\hati_h(t), \hati_v(t)) \approx e^{\tilde C_f(t- t^*) } (\hati_h(t^*), \hati_v(t^*)), \ \ t\ge t^*,
\end{align}
where $\tilde C_f = \begin{pmatrix} -\gamma_h &  \beta_h   \\C_0\beta_v  & -\gamma_v   \end{pmatrix}.$ It can easily be seen that when $\clr_0 < 1$, the matrix $\tilde C_f$ has two distinct negative eigenvalues $\lambda_{f,1}$, and $\lambda_{f,2}$. Thus from \eqref{me-2}, for $t\ge 0,$
\begin{align}\label{me-app-11}
e^{\tilde C_ft} =   [\mathbf{v}_{f,1}, \mathbf{v}_{f,2}] \begin{pmatrix} e^{\lambda_{f,1} t} & 0  \\ 0 & e^{\lambda_{f,2} t} \end{pmatrix}[\mathbf{v}_{f,1}, \mathbf{v}_{f,2}]^{-1},
\end{align}
where $\mathbf{v}_{f,1}, \mathbf{v}_{f,2}$ are the eigenvectors corresponding to $\lambda_{f,1}, \lambda_{f,2}$. From \eqref{approx-I} and \eqref{me-app-11}, it follows that when $\clr_0 < 1$, $(\hati_h(t), \hati_v(t))$ approaches $(0,0)$ exponentially fast as $t\to\infty$.

When $\clr_0 > 1$, and $(i_h(0), i_v(0))= E^e$, where $E^e$ is the endemic equilibrium point $E^e$ of the fluid limit $(i_h(t), i_v(t))$, which is globally asymptotically stable. Let $i_h^e, i_v^e$ denote the components of $E^e$, $s_h^e = 1 - i_h^e,$ and $s_v^e = C_0 - i_v^e$.
Then $(i_h(t), s_h(t), i_v(t), s_v(t)) =(i^e_h, s^e_h, i^e_v, s^e_v)$ for $t\ge 0.$ Now for $t\ge 0,$
\begin{align}
C(t)  \equiv C_e & =\begin{pmatrix} - (\gamma_h +\beta_h i_v^e) & 0 & \beta_h s_h^e \\ -\beta_vs_v^e&-\beta_v i_h^e &  \gamma_v \\ \beta_vs_v^e&\beta_v i_h^e &- \gamma_v \end{pmatrix} \label{rand-popl}\\ 
\sigma(t) \equiv \sigma_e &  = \begin{pmatrix} \sqrt{\beta_h i_v^e s_h^e+\gamma_h i_h^e}  & 0 & 0 & 0 \\ 0 & \sqrt{\gamma_v(C_0+s_v^e)} & - \sqrt{\beta_v i_h^es_v^e} & 0 \\ 0 & 0 & \sqrt{\beta_v i_h^e s_v^e} & \sqrt{\gamma_v i^e_v}\end{pmatrix}.
\end{align}
In this case, the solution of the SDE in \eqref{diff-sde} can be approximated by the following $3$-dimensional homogeneous linear SDE:
\begin{align} \label{diff-sde-r0g1}
\D \X(t) = C_e\X(t) \D t + \sigma_e \D \B(t).
\end{align}
Now note that $\X$ is a three-dimensional Ornstein-Uhlenbeck process. 
Solving the above SDE, we have the following approximation:
\begin{align}\label{sde_explicit_soln}
\X(t) = \exp\{C_et\} \X(0) +  \int_{0}^t \exp\{C_e(t-s)\} \sigma_e\D \B(s), \ t\ge 0.
\end{align}
{The covariance of $\X(t)$ then has the following approximation: }
\begin{equation}\label{cov-app-1}
\begin{aligned}
\var(\X(t))&  =  \var\left( \exp\{C_et\} \X(0)\right) + \EE\left[\left(\int_{0}^t \exp\{C_e(t-s)\} \sigma_e\D B(s)\right)^2\right]  \\
& = \exp\{C_et\} \var(\X(0)) \exp\{C_et\}' \\
& \quad + \int_{0}^t \exp\{C_e(t-s)\}\sigma_e{\sigma_e}'\exp\{C_e(t-s)\}' \D s. 
\end{aligned}
\end{equation}
We note that $C_e$ has three distinct eigenvalues, of which one is equal to $0$, and the other two are negative. Denote by $\lambda_{e,1},$ and $\lambda_{e,2}$ the two negative eigenvalues. From \eqref{me-2}, we have
\begin{align}\label{me-app-2}
e^{C_et}= [\mathbf{v}_{e,0}, \mathbf{v}_{e,1}, \mathbf{v}_{e,2}] \begin{pmatrix} 1 & 0 & 0 \\ 0 & e^{\lambda_{e,1} t}& 0\\ 0 & 0 & e^{\lambda_{e,2} t} \end{pmatrix}[\mathbf{v}_{e,0}, \mathbf{v}_{e,1}, \mathbf{v}_{e,2}]^{-1},
\end{align}
where $\mathbf{v}_{e,0}, \mathbf{v}_{e,1}, \mathbf{v}_{e,2}$ are the eigenvectors corresponding to $0, \lambda_1, \lambda_2$.
\begin{lemma}\label{infinite_var}
The variances of $\hati_h(t), \hats_v(t),$ and $\hati_v(t)$ approach $\infty$ as $t\to\infty$.
\end{lemma}
The above lemma implies that 
 $(\hati_h, \hats_v, \hati_v)$ has no quasi-stationary distribution, due to the fact that the vector population size is not fixed; rather at time $t$, $S^n_v(t)+I^n_v(t)$ is approximately normally distributed with mean $C_0n$ and variance $2n\gamma_vC_0t$ (see Remark \ref{rem-diff} (i)), which results in high variability of $(S^n_v(t), I^n_v(t))$ for large $t.$ However, when the vector population size is fixed, it can be shown that when $\clr_0 > 1$, $(I^n_h, I^n_v)$ admits a quasi-stationary distribution, which is approximately normal with mean $nE^e$ and variance $n\Sigma^*$ given in \eqref{stat-cov-fix} in Appendix \ref{vec-size-fixed}.

\begin{conjecture}\label{sum_case_I}
For large population size $n$,
\begin{itemize}
\item[\rm (i)] when $\clr_0 < 1$, we get from \eqref{approx-I} and Theorem \ref{lln-stability-I} (i) that $(I^n_h(t), I^n_v(t))$ approaches $0$ exponentially fast as $t\to\infty$;
\item[\rm (ii)] when $\clr_0 > 1$, for the epidemic model with fixed host population size and random vector population size defined in \eqref{sys-1} -- \eqref{sys-4}, $(I^n_h, I^n_v)$ has no quasi-stationary distribution.
\item[\rm (iii)] when $\clr_0 >1$, for the epidemic model with fixed host and vector population sizes defined in \eqref{sys-11} -- \eqref{sys-14}, $(I^n_h, I^n_v)$ admits a quasi-stationary distribution, which can be approximated b a normal distribution with mean $nE^e$ and variance $n\Sigma^*$ given in \eqref{stat-cov-fix} in Appendix \ref{vec-size-fixed}.
\end{itemize}
\end{conjecture}

\section{Numerical experiments}\label{sec:numerical}
In this section, we conduct numerical experiments to validate the approximations established for quasi-stationary distributions. 

\begin{example} Let $\clr_0 > 1$ and consider the scaling Case I. In this example, we compare the simulated second moments of $\hati^n_h$, $\hats^n_v$, and $\hati^n_v$ with the second moments of $\hati_h, \hats_v,$ and $\hati_v$, respectively, and observe the behaviors of the moments as $t\to\infty$. We set the initial value of $\{(\bar I^n_h(t), \bar I^n_v(t)); t\ge 0\}$ to be endemic equilibrium point $$E^e = \left(\frac{C_0\beta_h\beta_v - \gamma_h\gamma_v}{C_0\beta_h\beta_v + \beta_v\gamma_h}, \frac{C_0\beta_h\beta_v - \gamma_h\gamma_v}{\beta_h\beta_v + \beta_h\gamma_v}\right).$$ 
Noting that $E^e$ is globally asymptotically stable, the fluid limit $(i_h(t), i_v(t)) = E^e$ for all $t\ge 0.$ 

We let $n=10000, \beta_h^n=0.2, \beta_v^n=0.1, \gamma_h^n=0.3, \gamma_v^n=0.1$, and  $C_0 = 5$. 
\begin{itemize}
\item[\rm (i)] We consider the vector-borne SIS model defined in \eqref{sys-1} -- \eqref{sys-4}. We generate $50$ sample paths, and calculate the simulated second moments of $\hati^n_h$, $\hats^n_v$, and $\hati^n_v$. 

From the analysis in Section \ref{quasi-case-1}, $(\hati_h(t), \hats_v(t), \hati_v(t))$ satisfies the SDE \eqref{diff-sde-r0g1}, and its covariance matrix is given in \eqref{cov-app-1}. We also note that $(\hati_h(0), \hats_v(0), \hati_v(0)) = 0$, and from \eqref{sde_explicit_soln}, $\EE[(\hati_h(t), \hats_v(t), \hati_v(t))] = (0,0,0).$ So the second moments of $\hati_h$, $\hats_v$, and $\hati_v$ equal to their variances. 

From Figure \ref{f:Ex5_1}, we see that the variances in \eqref{cov-app-1} provide nice approximations for the second moments of the pre-limits, and all variances increase as the time $t$ increases, which valides the conclusion that the system has no quasi-stationary (see Summary \ref{sum_case_I} (ii)). 
\item[\rm (ii)] We consider the vector-borne SIS model defined in \eqref{sys-11} -- \eqref{sys-14}. We generate $50$ sample paths, and calculate the simulated second moments of $\hati^n_h$, and $\hati^n_v$. 

From Appendix \ref{vec-size-fixed}, $(\hati_h, \hati_v)$ satisfies the SDE \eqref{sde_fix_popu}, and its covariance matrix is given in \eqref{cov-app-2}. From \eqref{sde_explicit_soln_fix_popu}, it follows that $\EE[(\hati_h(t), \hati_v(t))] = (0,0)$, and the second moments of $\hati_h$ and $\hati_v$ equal to their variances.   

From Figure \ref{f:Ex5_2}, we see that the variances in \eqref{cov-app-2} provide nice approximations for the second moments of the pre-limits, and all variances become stable as the time $t$ increases, which valides the conclusion that the system has a quasi-stationary (see Summary \ref{sum_case_I} (iii)). 
\end{itemize}
\begin{figure}[!htb]
 \centerline{\includegraphics[width=150mm]{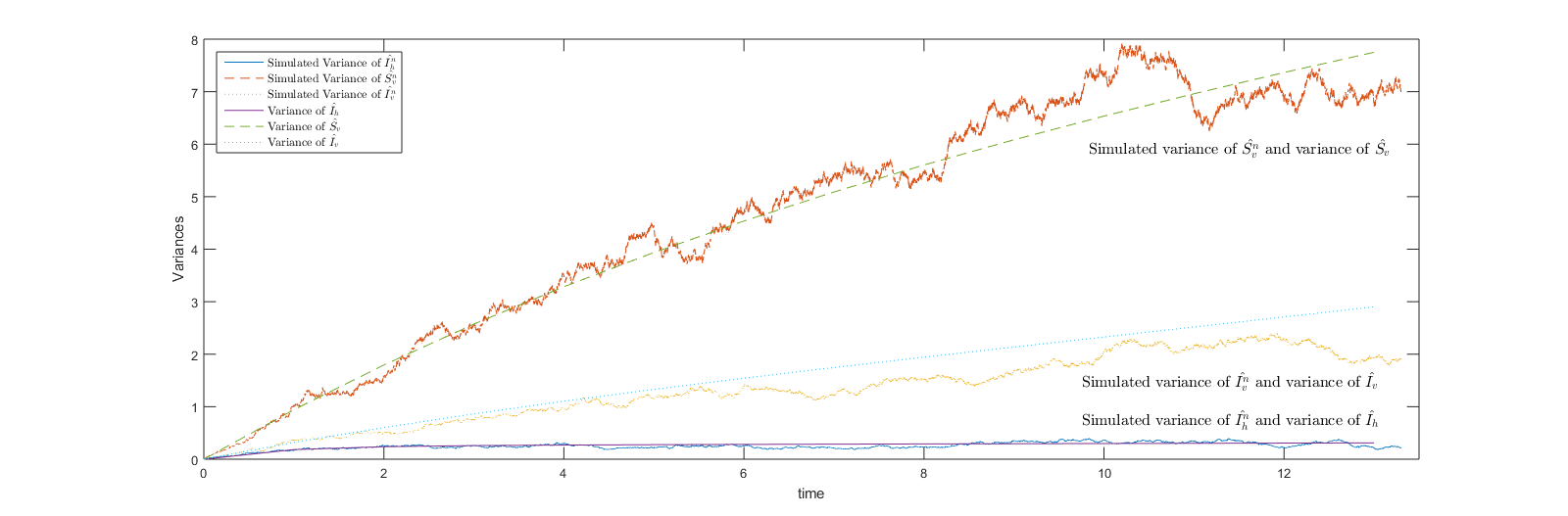}}
\caption{Simulated second moments of $\hati^n_h$, $\hats^n_v$, and $\hati^n_v$, and the variance of $\hat I_h, \hat S_v$, and $\hat I_v$ for the model \eqref{sys-1} -- \eqref{sys-4}. }
\label{f:Ex5_1}
\end{figure}
\begin{figure}[!htb]
 \centerline{\includegraphics[width=150mm]{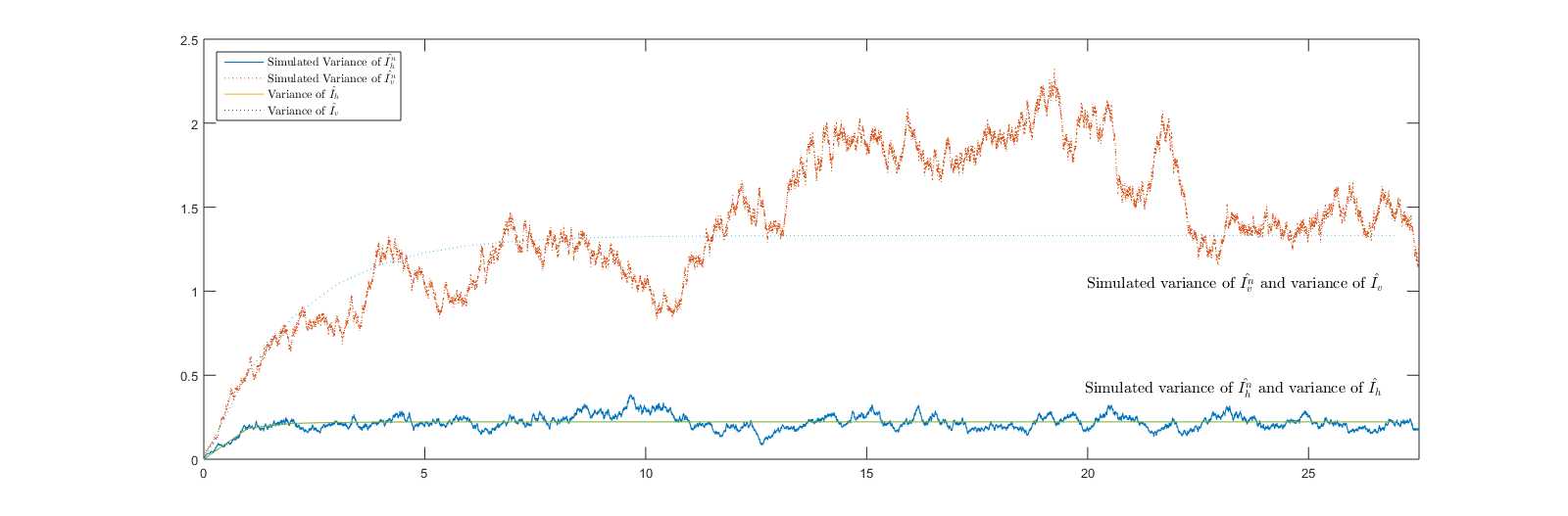}}
\caption{Simulated second moments of $\hati^n_h$, and $\hati^n_v$, and the variance of $\hat I_h$, and $\hat I_v$ for the model \eqref{sys-1} -- \eqref{sys-4}. }
\label{f:Ex5_2}
\end{figure}
\end{example}

\begin{example} Let $\clr_0 > 1$, and consider the scaling Case II. We study the behavior of the fluid scaled processes in the system defined in \eqref{sys-1} -- \eqref{sys-4}. Let $n=100$, $\beta_h^n=0.2$, $\gamma_h^n=0.3$, $\beta_v^n=0.1n^{\frac{2}{3}}$, $\gamma_v^n=0.2n^{\frac{2}{3}}$ and $C_0=5$. Set $(\bar I^n_h(0), \bar I^n_v(0))=\left(\frac{C_0\beta_h\beta_v - \gamma_h\gamma_v+1}{C_0\beta_h\beta_v + \beta_v\gamma_h}, \frac{C_0\beta_h\beta_v - \gamma_h\gamma_v+1}{\beta_h\beta_v + \beta_h\gamma_v}\right)$, which is not equal to either equilibrium point. We simulate $10$ sample paths of $\bar I^n_h$, and plot them together with the fluid limit in Figure \ref{f:Ex5_3}.
\begin{figure}[!htb]
 \centerline{\includegraphics[width=150mm]{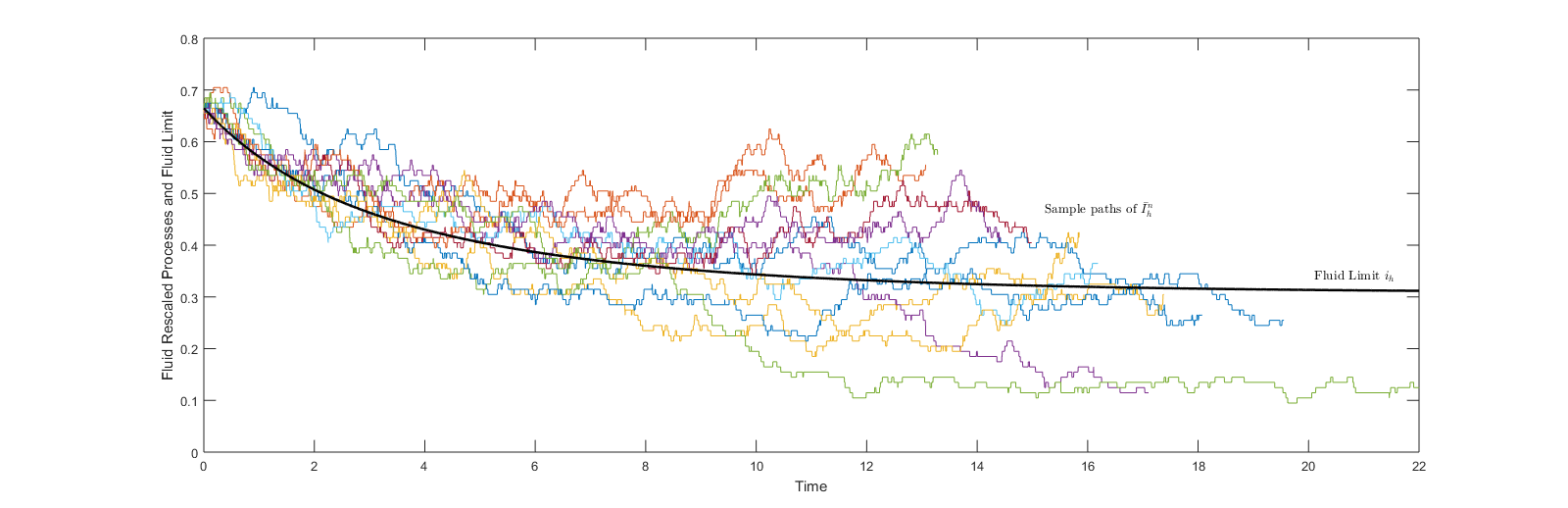}}
\caption{Fluid scaled Processes $\bar I^n_h$ and its fluid Limit $i_h$ under scaling Case II.}
\label{f:Ex5_3}
\end{figure}

\end{example}

\section{Discussion}\label{sec:disc}

Epidemic processes are essentially stochastic, but analyses of these stochastic models are often far from straightforward. This paper provides novel analysis approaches of simple stochastic epidemic models for vector-borne infectious diseases. Our approach uses several probabilistic and statistical techniques to reduce the dimension of the system and develop important mathematical quantities for understanding disease outbreaks and persistence. Rather than focusing on any specific disease, we instead rigorously analyzed a simple model and introduced several techniques that can be potentially used to study more general and complex disease models. Techniques that are explained here include martingales, scaling limit theorems, and quasi-stationary distributions. Specifically, the vector-borne epidemic model is formulated as a multi-dimensional CTMC. Using the fluid and diffusion approximations for CTMCs, efficient approximations of quasi-stationary distributions are provided.

\bibliographystyle{apa}
\bibliography{stoch_epidemic_ref}{}

\newpage

\appendix

\section{Appendix: Equilibrium points of Differential Equations}\label{app:DE}
Consider a differential equation:
\begin{align}\label{DE}
y'(t) = f(y),
\end{align}
where $f: \RR^K \to \RR^K$, and $y\in \RR^K.$ A point $y_0\in\RR^K$ is called an {\em equilibrium point} of \eqref{DE} if $f(y_0) =0.$ An equilibrium point $y_0$ is said to be {\em stable} if given $\epsilon > 0$ there is a $\delta > 0$ such that for every solution $y$ of \eqref{DE}, when $|y(0)-y_0| < \delta$, we have $\|y(t) - y_0\| < \epsilon$ for all $t >0$; it is said to be {\em locally asymptotically stable (LAS)} if there exists $a>0$ such that if $\|y(0)-y_0\|\le a$, then $y(t)\to y_0$ as $t\to\infty$; it is said to be {\em globally asymptotically stable (GAS)} if $y(t)\to y_0$ as $t\to\infty$; it is {\em globally exponentially stable (GES)} if $y_0$ is GAS and there exists $M, \kappa >0$ such that $\|y(t) - y_0\| \le M e^{-\kappa t}$; it is said to be {\em unstable} if it is not stable. (See Chapter 1 in \cite{perko2013differential} for more detail.) 

\section{Appendix: Matrix exponentials}\label{appendix:me}

The matrix exponential is a matrix function on square matrices. For a $n\times n$ square matrix $A$, the exponential of $A$, denoted by $e^A$ or $\exp\{A\}$, is defined as
\[
e^A = \sum_{k=0}^\infty \frac{A^k}{k!}.
\]
It is easily seen that, for the special case when $n=1$, the matrix exponential is reduced to be the regular exponential.

For $k\times k$ square matrices with $k$ distinct real eigenvalues $\lambda_1, \ldots, \lambda_k$, we have
\begin{align}
e^{At} &  = [\mathbf{v}_1, \ldots, \mathbf{v}_k] \begin{pmatrix} e^{\lambda_1 t} & \cdots & 0  \\ \vdots & \vdots & \vdots \\ 0& \cdots & e^{\lambda_k t} \end{pmatrix}[\mathbf{v}_1, \ldots, \mathbf{v}_k]^{-1},\label{me-2}
\end{align}
where $\mathbf{v}_1, \ldots, \mathbf{v}_k$ are the eigenvectors corresponding to $\lambda_1, \ldots, \lambda_k$. (See Chapter 1 in \cite{perko2013differential} for more detail.)


\section{Appendix: Vector-borne SIS with fixed vector population size}\label{vec-size-fixed}
We consider a stochastic vector-borne SIS model, in which the population sizes of hosts and vectors are assumed to be $n$ and $C_0 n$, respectively, where $C_0$ is some positive constant. The system equations can by formulated by using independent unit rate Poisson processes $N^n_i, i=1, 2, 3, 4$. For $t\ge 0,$ we have
\begin{align}
S_h^n(t) & = S_h^n(0)  - N_1^n\left(\beta_h^n\int_0^t \frac{I_v^n(u) S_h^n(u)}{n} \D u \right)  + N_2^n\left(\gamma_h^n \int_0^t  I_h^n(u) \D u \right), \label{sys-11}\\
I_h^n(t) & = I_h^n(0) + N_1^n\left(\beta_h^n \int_0^t \frac{I_v^n(u) S_h^n(u)}{n} \D u \right)  - N_2^n\left(\gamma_h^n \int_0^t  I_h^n(u) \D u \right), \label{sys-12}
\end{align}
and
\begin{align}
S_v^n(t) & = S_v^n(0) - N_3^n\left(\beta_v^n  \int_0^{ t} \frac{I_h^n(u) S_v^n(u)}{n} \D u \right)+ N_4^n \left(\gamma_v^n \int_0^t  I_v^n(u) \D u \right), \label{sys-13}\\
I_v^n(t) & = I_v^n(0) + N_3^n\left(\beta_v^n  \int_0^{ t} \frac{I_h^n(u) S_v^n(u)}{n} \D u \right)- N_4^n \left(\gamma_v^n \int_0^t  I_v^n(u) \D u \right). \label{sys-14}
\end{align}
We note that the epidemic system can be described by the $2$-dimensional CTMC $(I^n_h, I^n_v)$ with a finite state space $\{0, 1, \ldots, n\}\times \{0, 1, \ldots, C_0n\}$. Furthermore,  $(I^n_h, I^n_v)$ has an absorbing state $(0,0)$, and a unique stationary distribution $\delta_{(0,0)}.$

We still consider the two scaling cases in Section \ref{sec:scaling}, and the reproducation number is the same as in Section \ref{sec:reprod}. The fluid and diffusion approximations can be established similar to Theorems \ref{lln-1} -- \ref{lln-stability-II}, \ref{fclt-1}, and \ref{fclt-2}. In particular, for both cases, the fluid limits of $(I^n_h, I^n_v)$ are the same as Theorems \ref{lln-1} -- \ref{lln-stability-II}. For Case II, the diffusion limit of $\hati^n_h$ is the same as Theorem \ref{fclt-2}. In the following, we give the diffusion approximation under Case I.

\begin{theorem}\label{fclt-11} Consider Case I, and assume that $(\hat I^n_h(0), \hat I^n_v(0))$ converges in distribution to some random variable $(\hat I_h(0), \hat I_v(0))$, and that
\begin{align*}
\sqrt{n}(\beta^n_h - \beta_h) \to \hat\beta_h, \ \ \sqrt{n}(\gamma^n_h - \gamma_h) \to \hat\gamma_h, \ \ \sqrt{n}(\beta^n_v - \beta_v) \to \hat\beta_v, \ \ \sqrt{n}(\gamma^n_v - \gamma_v) \to \hat\gamma_v.
\end{align*}
 Then we have that $(\hat I_h^n, \hat I_v^n)$ converges in distribution to $(\hat I_h, \hat I_v)$, where $(\hati_h, \hati_v)$ is the unique solution to the following stochastic integral equations: For $t\ge 0,$
\begin{align*}
\hat I_h(t) & = \hat I_h(0) + \int_0^t \beta_h s_h(u)\hat I_v(u) - (\gamma_h +\beta_h i_v(u))\hat I_h(u) + \hat\beta_h i_v(u) s_h(u) - \hat\gamma_h i_h(u) \ \D u \\
& \quad + \int_0^t \sqrt{\beta_h i_v(u) s_h(u)+\gamma_h i_h(u)} \  \D B_1(u), \\
\hat I_v(t) & = \hat I_v(0) + \int_0^t \beta_vs_v(u)\hat I_u(u) - \left(\gamma_v +\beta_v i_h(u) \right)\hat I_v(u) + \hat\beta_v i_h(u)s_v(u)- \hat\gamma_v i_v(u) \  \D u \\
& \quad + \int_0^t \sqrt{\beta_v i_h(u)s_v(u)+\gamma_v i_v(u)} \ \D  B_2(u),
 \end{align*}
 with $B_1$ and $B_2$ being two independent standard Brownian motions.
\end{theorem}

We next study the quasi-stationary distribution of $(I^n_h, I^n_v)$ when $\clr_0 > 1$. Using the fluid and diffusion approximations, we have  for large $n$,
\begin{align*}
(I^n_h(t), I_v^n(t)) \approx  \underbrace{n(i_h(t), i_v(t))}_{\mbox{Fluid approximation}} +  \underbrace{\sqrt{n}(\hati_h(t), \hati_v(t)),}_{\mbox{Diffusion approximation}} \  \mbox{in distribution.}
\end{align*}

For $t\ge 0,$ let
\begin{align} \label{appdxC}
C(t) & = \begin{pmatrix} - (\gamma_h +\beta_h i_v(t)) & \beta_h s_h(t) \\ \beta_vs_v(t) &- \left(\gamma_v +\beta_v i_h(t) \right)  \end{pmatrix}, \\
\label{appdxsigma} \sigma(t) & = \begin{pmatrix} \sqrt{\beta_h i_v(t) s_h(t)+\gamma_h i_h(t)}  & 0 \\ 0 & \sqrt{\beta_v i_h(t)s_v(t)+\gamma_v i_v(t)}\end{pmatrix},\\
\B(t) & = \begin{pmatrix} B_1(t) \\ B_2(t) \end{pmatrix},
\end{align}
and
\[
\hati(t) = (\hati_h(t), \hati_v(t)).
\]
Then the diffusion limit process $\hati$ satisfies the following $2$-dimensional linear SDE:
\begin{align}\label{diff-sde-1}
\D \hati(t) & = C(t)\hati(t) \D t + \sigma(t) \D \B(t).
\end{align}

When $\clr_0 > 1$ and $(i_h(0), i_v(0)) = E^e$, where $E^e$ is the globally asymptotically stable endemic equilibrium point of $(i_h(t), i_v(t))$. Then for $t\ge 0$, $(i_u(t), i_v(t))= E^e$, and it follows that
\begin{align}
C(t) & \equiv C_e =\begin{pmatrix} - (\gamma_h +\beta_h i_v^e) & \beta_h s_h^e \\ \beta_vs_v^e&- \left(\gamma_v +\beta_v i_h^e \right)   \end{pmatrix}, \label{fixed-popl}\\ 
\sigma(t) & \equiv \sigma_e = \begin{pmatrix} \sqrt{\beta_h i_v^e s_h^e+\gamma_h i_h^e}  & 0 \\ 0 & \sqrt{\beta_v i_h^es_v^e+\gamma_v i_v^e}\end{pmatrix}. 
\end{align}
The SDE in \eqref{diff-sde-1} can be simplified as the following $2$-dimensional homogeneous linear SDE
\begin{align}\label{sde_fix_popu}
\D \hat I(t) = C_e\hat I(t) \D t + \sigma_e \D \B(t).
\end{align}
Solving the above SDE, we have 
\begin{align}\label{sde_explicit_soln_fix_popu}
\hati(t) = \exp\{C_et\} \hati(0) +  \int_{0}^t \exp\{C_e(t-s)\} \sigma_e\D B(s), \ t\ge 0,
\end{align}
whose covariance matrix is given as 
\begin{align}\label{cov-app-2}
\var(\hati(t)) = \exp\{C_et\} \var(\hati(0)) \exp\{C_et\}' + \int_0^t  \exp\{C_e(t-s)\}\sigma_e{\sigma_e}'\exp\{C_e(t-s)\}' \D s. 
\end{align}
It can be seen that when $\clr_0 > 1$, the matrix $C_e$ has two distinct negative eigenvalues $\lambda_1$ and $\lambda_2$. Thus, from \eqref{me-2}, we have
\begin{align}
\Sigma^* \equiv \lim_{t\to\infty} \var(\hati(t)) 
 = \int_0^\infty \exp\{C_es\}\sigma_e{\sigma_e}'\exp\{C_es\}' \D s < \infty. \label{stat-cov-fix}
\end{align}

\section{Appendix: Proofs}\label{proof}

We provide all the proofs in this section. The Poisson processes in our model \eqref{sys-1} -- \eqref{sys-4} depend on the parameter $n$, and the LLN and CLT results from \cite{kurtz1978strong} or Chapter 11 of \cite{Ethier:Kurtz:1986} cannot be applied directly. To prove convergence in distribution for the sequences of stochastic processes in Theorems \ref{lln-1}, \ref{lln-2}, \ref{fclt-1} and \ref{fclt-2}, we first show the $C$-tightness of the sequences, and then characterize the uniqueness of the weak limits. 

Let $(\Omega, \mathcal{F},\mathbb{P})$ be a complete probability space; all random variables and stochastic processes described in this work are, without loss of generality, defined on this common probability space. The following notation will be used. Let $\RR^K_+ = \{x\in \RR^K: x_i \ge 0, i=1, \ldots, K\}$. Denote by $D([0,\infty), \RR^K_+)$ the space of right continuous functions with left limits (RCLL) from $[0,\infty)$ to $\RR^K_+$ equipped with the usual Skorohod topology, and  $\mathcal{C}^2_0(\RR^K_+)$ the space of twice differentiable bounded functions from $\RR^K_+ \to \RR$. A stochastic process $X$ with values in $\RR^K$ will be regarded as a random variable with values in $D([0,\infty), \RR^K)$. A sequence of RCLL stochastic processes $\{X^n\}_{n\ge 1}$ is said to be $C$-tight if $\{X^n\}_{n\ge 1}$ is tight and any weak limit has continuous sample paths. Convergence in distribution of random variables/stochastic processes $X^n$ to $X$ will be denoted as $X^n \Rightarrow X$.

The following theorem from \cite{Billingsley99} will be used in the proofs.
\begin{theorem}[\cite{Billingsley99}]\label{tightness}
The sequence of stochastic processes $\{X^n(t); t\ge 0\}$ in $D([0,\infty), \RR^K)$ is $C$-tight if and only if the following two conditions hold:
\begin{itemize}
\item[\rm (i)] For any $T\ge 0$, \[\lim_{a\uparrow\infty}\PP\left(\sup_{0\le t \le T} |X^n(t)|> a\right) =0, \ \ n\ge 1. \]
\item[\rm (ii)] For any $\epsilon>0$ and $0\le t_1 \le t_2 < \infty$,
\[
\lim_{\delta\downarrow 0}\limsup_{n\to\infty} \PP\left( \sup_{0\le t_1 \le t_2 \le t_1+\delta} |X^n(t_2) - X^n(t_1)|>\epsilon\right) =0.
\]
\end{itemize}
\end{theorem}

\subsection{Proofs for Case I}

In this section, we prove Theorems \ref{lln-1}, \ref{lln-stability-I}, and \ref{fclt-1} for scaling Case I. 
\begin{proof}[Proof of Theorem \ref{lln-1}]
For the unit-rate Poisson processes $N_i^n$, define
\begin{align*}
\bar N_i^n(t)= \frac{1}{n}N^n_i(nt), \ \mbox{and} \ \bar N_i^{n,c}(t) = \frac{1}{n} [N_i^n(nt) - nt].
\end{align*}
We first show the $C$-tightness of $(\bar I^n_h, \bar S^n_v, \bar I^n_v)$. We observe that for $t\ge 0$,
\begin{align*}
\bar I^n_h(t) & = \bar I^n_h(0) + \bar N^n_1\left( \beta^n_h\int_0^t \bar I^n_v(u)\bar S^n_h(u)\D u\right) - \bar N^n_2\left(\gamma^n_h\int_0^t \bar I^n_h(u) \D u \right), \\
\bar S^n_v(t) & = \bar S^n_v(0) + \bar N^n_3\left(\gamma^n_v\int_0^t \bar S^n_v(u) + \bar I^n_v(u) \D u\right) - \bar N^n_4\left(\beta^n_v\int_0^t \bar I^n_h(u) \bar S^n_v(u) \D u\right)\\
& \quad  - \bar N^n_5\left(\gamma^n_v\int_0^t \bar S^n_v(u) \D u\right), \\
\bar I^n_v(t) & = \bar I^n_v(0) + \bar N^n_4\left(\beta^n_v\int_0^t \bar I^n_h(u) \bar S^n_v(u) \D u\right) - \bar N^n_6\left(\gamma^n_v\int_0^t \bar I^n_v(u) \D u\right).
\end{align*}
It suffices to show that the two conditions in Theorem \ref{tightness} hold for $(\bar I^n_h, \bar S^n_v, \bar I^n_v)$. Noting that $\bar I^n_h(t)\in [0, 1]$ for all $t\ge 0$, we consider $(\bar S^n_v, \bar I^n_v)$, and observe that for $t\ge 0$,
\begin{align}
\bar S^n_v(t) + \bar I^n_v(t) & = C_0 + \bar N^n_3\left(\gamma^n_v\int_0^t \bar S^n_v(u) + \bar I^n_v(u) \D u\right) -\bar N^n_5\left(\gamma^n_v\int_0^t \bar S^n_v(u) \D u\right) \nonumber\\
& \quad - \bar N^n_6\left(\gamma^n_v\int_0^t \bar I^n_v(u) \D u\right) \label{eqn:vectorpopu}\\
& \le C_0 + \bar N^n_3\left(\gamma^n_v\int_0^t \bar S^n_v(u) + \bar I^n_v(u) \D u\right) +\bar N^n_5\left(\gamma^n_v\int_0^t \bar S^n_v(u) \D u\right) \nonumber\\
& \quad + \bar N^n_6\left(\gamma^n_v\int_0^t \bar I^n_v(u) \D u\right).\label{ieqn:vectorpopu}
\end{align}
It follows, from \eqref{eqn:vectorpopu}, that
\begin{align}\label{meanvecsize}
\EE[\bar S^n_v(t) + \bar I^n_v(t)]  = C_0,
\end{align}
and from \eqref{ieqn:vectorpopu}, for $T\ge 0,$
\begin{align*}
\EE\left[\sup_{0\le t\le T}(\bar S^n_v(t) + \bar I^n_v(t))\right]& \le C_0 + 2\gamma^n_v \int_0^T \EE[\bar S^n_v(u) + \bar I^n_v(u)] \D u = C_0 + 2\gamma^n_v C_0 T,
\end{align*}
which implies that the condition in Theorem \ref{tightness} (i) holds.
Next we show that the condition in Theorem \ref{tightness} (ii) holds.
We note that for $\delta >0,$
\begin{align*}
\sup_{|t-s|<\delta} |(\bar I^n_h(t), \bar S^n_v(t), \bar I^n_v(t)) - (\bar I^n_h(s), \bar S^n_v(s), \bar I^n_v(s))| \le {\rm (I) + (II) + (III) +2 (IV) + (V) + (VI)},
\end{align*}
where
\begin{align*}
{\rm (I)} & = \sup_{|t-s|<\delta} \left| \bar N^n_1\left( \beta^n_h\int_0^t \bar I^n_v(u)\bar S^n_h(u)\D u\right) - \bar N^n_1\left( \beta^n_h\int_0^s \bar I^n_v(u)\bar S^n_h(u)\D u\right) \right|, \\
{\rm (II)} & = \sup_{|t-s|<\delta} \left| \bar N^n_2\left(\gamma^n_h\int_0^t \bar I^n_h(u) \D u \right)- \bar N^n_2\left(\gamma^n_h\int_0^s \bar I^n_h(u) \D u \right) \right|, \\
{\rm (III)} & = \sup_{|t-s|<\delta} \left| \bar N^n_3\left(\gamma^n_v\int_0^t \bar S^n_v(u) + \bar I^n_v(u) \D u\right)-  \bar N^n_3\left(\gamma^n_v\int_0^s \bar S^n_v(u) + \bar I^n_v(u) \D u\right) \right|, \\
{\rm (IV)} & = \sup_{|t-s|<\delta} \left| \bar N^n_4\left(\beta^n_v\int_0^t \bar I^n_h(u) \bar S^n_v(u) \D u\right)- \bar N^n_4\left(\beta^n_v\int_0^s \bar I^n_h(u) \bar S^n_v(u) \D u\right) \right|, \\
{\rm (V)} & = \sup_{|t-s|<\delta} \left| \bar N^n_5\left(\gamma^n_v\int_0^t \bar S^n_v(u) \D u\right)- \bar N^n_5\left(\gamma^n_v\int_0^s \bar S^n_v(u) \D u\right) \right|, \\
{\rm (VI)} & = \sup_{|t-s|<\delta} \left| \bar N^n_6\left(\gamma^n_v\int_0^t \bar I^n_v(u) \D u\right)- \bar N^n_6\left(\gamma^n_v\int_0^s \bar I^n_v(u) \D u\right)\right|.
\end{align*}
For (II), noting that $0\le \sup_{t\ge 0} \bar I^n_h(t) \le 1$, it follows that
\begin{align}
\EE[{\rm (II)}] & \le \EE\left[ \sup_{|t-s|<\delta} \sup_{0\le \lambda \le \gamma^n_h} \left| \bar N^n_2\left(\lambda t \right)- \bar N^n_2\left(\lambda s \right) \right|\right] = \EE\left[ \sup_{|t-s|<\delta} \sup_{0\le \lambda \le \gamma^n}  \bar N^n_2\left(\lambda |t-s| \right)\right] \nonumber\\
& \le \EE[\bar N^n_2\left(\gamma^n_h \delta \right)] = \gamma^n_h \delta. \label{est-II}
\end{align}
To continue, we note that from functional law of large numbers for Poisson processes, $\bar N^{n}_i$ converges weakly to the identity map from $[0,\infty)$ to $[0,\infty)$, and thus it is certainly tight.
Now for (III), we have
\begin{align}
&\lim_{\delta\downarrow 0} \limsup_{n\to\infty}\PP({\rm (III)}>\epsilon) \nonumber\\
& = \lim_{K\uparrow\infty} \lim_{\delta\downarrow 0} \limsup_{n\to\infty}\PP\left({\rm (III)}>\epsilon, \sup_{0\le u \le \max\{t, s\}+\delta}[\bar S^n_v(t) + \bar I^n_v(t)] \le K\right)\nonumber \\
& \quad + \lim_{K\uparrow\infty} \lim_{\delta\downarrow 0} \limsup_{n\to\infty}\PP\left({\rm (III)}>\epsilon, \sup_{0\le u \le \max\{t, s\}+\delta}[\bar S^n_v(t) + \bar I^n_v(t)] > K\right) \nonumber \\
& \le \lim_{K\uparrow\infty} \lim_{\delta\downarrow 0} \limsup_{n\to\infty}\PP\left(\sup_{|t-s|<\delta} \sup_{0\le \lambda \le \gamma^n_vK}  \bar N^n_3\left(\lambda |t-s| \right)>\epsilon\right)\nonumber \\
& \quad + \lim_{K\uparrow\infty}\lim_{\delta\downarrow 0}  \limsup_{n\to\infty}\PP\left(\sup_{0\le u \le \max\{t, s\}+\delta}[\bar S^n_v(t) + \bar I^n_v(t)] > K\right) \nonumber \\
& \le \lim_{K\uparrow\infty} \lim_{\delta\downarrow 0} \limsup_{n\to\infty}\frac{\EE[\bar N^n_3\left(\gamma^n_v K\delta \right)]}{\epsilon}  + \lim_{K\uparrow\infty}\lim_{\delta\downarrow 0}  \limsup_{n\to\infty}\frac{\EE[\sup_{0\le u \le \max\{t, s\}+\delta}(\bar S^n_v(t) + \bar I^n_v(t))]}{K} \nonumber  \\
& = 0. \label{est-III}
\end{align}
Using similar arguments to those for (III), we can show that
\begin{align}
\lim_{\delta\downarrow 0} \limsup_{n\to\infty}\PP({\rm (I)}>\epsilon) & =0, \label{est-I} \\
\lim_{\delta\downarrow 0} \limsup_{n\to\infty}\PP({\rm (IV)}>\epsilon) & =0, \label{est-IV} \\
\lim_{\delta\downarrow 0} \limsup_{n\to\infty}\PP({\rm (V)}>\epsilon) & =0, \label{est-V} \\
\lim_{\delta\downarrow 0} \limsup_{n\to\infty}\PP({\rm (VI)}>\epsilon) & =0. \label{est-VI}
\end{align}
The desired condition in Theorem \ref{tightness} (ii) follows from \eqref{est-II} -- \eqref{est-VI}.

 Let $(i_h, s_v, i_v)$ be a weak limit. Using the fact that $\bar N^n_i$ converges to the identity map from $[0,\infty)$ to $[0,\infty)$, we see that $(i_h, i_v)$ satisfies the ODE system defined by \eqref{ode1-1} and \eqref{ode2-1}, and $s_v(t) = C_0 -i_v(t)$ for $t\ge 0$. Finally, since the ODEs \eqref{ode1-1} and \eqref{ode2-1} have a unique solution, we conclude that $(\bar I^n_h, \bar S^n_v, \bar I^n_v)$ converges to $(i_h, s_v, i_v)$ weakly. Lastly, we note that from \eqref{eqn:vectorpopu}, for $t\ge 0$,
\begin{align*}
\EE\left[\sup_{0\le t\le T}(\bar S^n_v(t) + \bar I^n_v(t))^2\right]& \le 4 C_0^2 + 4 \EE\left[\bar N^n_3\left(\gamma^n_v\int_0^T \bar S^n_v(u) + \bar I^n_v(u) \D u\right)\right]^2 \\
& \quad + 4 \EE\left[\bar N^n_5\left(\gamma^n_v\int_0^T \bar S^n_v(u)\D u\right)\right]^2 + 4 \EE\left[\bar N^n_6\left(\gamma^n_v\int_0^T \bar I^n_v(u) \D u\right)\right]^2 \\
& = 4 C_0^2 + 8 \gamma^n_v\int_0^T \EE[\bar S^n_v(u) + \bar I^n_v(u)] \D u  \\
& = 4 C_0^2 + 8 \gamma^n_v C_0 T,
\end{align*}
which implies the uniform integrability of $(\bar S^n_v, \bar I^n_v)$. It now follows that $\EE[|(\bar I^n_h, \bar S^n_v, \bar I^n_v)-(i_h, s_v, i_v)|] \to 0.$

\end{proof}

\begin{proof}[Proof of Theorem \ref{lln-stability-I}]
We investigate the properties of the ODE system \eqref{ode1-1} and \eqref{ode2-1}. Let $\Omega = [0,1]\times [0, C_0]$. It can be easily verified that the set $\Omega$ is positively invariant for the system \eqref{ode1-1} and \eqref{ode2-1}, which says when $(i_h(0), i_v(0))\in \Omega$ then $(i_h(t), i_v(t))\in \Omega$ for any $t\ge 0.$ For $(x,y)\in \Omega$, let
\begin{align}
f(x,y) & = \beta_h y (1-x) -\gamma_h x, \\
g(x,y)& = \beta_v x (C_0 - y) - \gamma_v y.
\end{align}
Then
\[
i_h'(t) = f(i_h, i_v), \ \ i_v'(t) = g(i_h, i_v).
\]
Noting that $f$ and $g$ are Lipschitz continous, the ODE system \eqref{ode1-1} and \eqref{ode2-1} has a unique solution. The equilibrium points $E^f$ and $E^e$ can be obtained by letting $f(i_h, i_v)$ and $g(i_h, i_v)$ equal to $0$, and solving the equations for $(i_u, i_v)$. We next show the stability properties of $E^f$ and $E^e$. First, the locally asymptotic stability can be shown by linearizing $f$ and $g$ at the equalibrium point, and observing the eigenvalues of the corresponding Jacobian matrix. For the globally asymptotic stability of $E^f$. We consider the Lyapunov function $L(i_h,i_v) = ai_h + bi_v$, where $a, b$ are positive constants such that $C_0\beta_v/\gamma_h \le a/b \le \gamma_v/\beta_h$. Then $L$ is positive definite, and $\dot{L}(i_h, i_v) =  L_x(i_h(t), i_v(t))i_h'(t) + L_y(i_h(t), i_v(t)) i_v'(t)$ is nonpositive, and is equal to $0$ only when $(i_h, i_v) = E^f$. From Lasalle's theorem, $E^f$ is globally asymptotically stable when $\mathcal{R}_0 \le 1.$ When $\clr_0 <1$, $a$ and $b$ can be chosen such that $C_0\beta_v/\gamma_h < a/b < \gamma_v/\beta_h$, and then
\begin{align*}
\dot{L}(i_h, i_v) & =  L_x(i_h(t), i_v(t))i_h'(t) + L_y(i_h(t), i_v(t)) i_v'(t)\\
& \le (bC_0 \beta_v  - a \gamma_h)i_h(t) +  (a\beta_h - b\gamma_v) i_v(t) \\
& \le -\alpha L(i_h(t), i_v(t)),
\end{align*}
where $\alpha>0$ satisfies $(C_0\beta_v + \alpha)(\beta_h +\alpha) \le \gamma_h\gamma_v.$ Hence $E^f$ is exponentially asymptotically stable when $\clr_0 <1$.  The global asymptotical stability of $E^e$ follows from Theorem 2.1 in \cite{BERETTA1986677}. Adapting the notation from \cite{BERETTA1986677}, for our system \eqref{ode1-1} and \eqref{ode2-1}, let \[A = \begin{pmatrix} 0 & -\beta_h \\ -C_0\beta_v& 0\end{pmatrix}, e = \begin{pmatrix} -\gamma_h\\-\gamma_v \end{pmatrix}, B = \begin{pmatrix} 0&\beta_h \\ C_0\beta_v& 0\end{pmatrix}, c = \begin{pmatrix} 0 \\ 0 \end{pmatrix}.\]
It can be verified that when $W_1, W_2 >0$ satisfy
\[
W_1 \beta_h \left(\frac{1}{i_h^*} -1 \right) = W_2C_0 \beta_v \left(\frac{1}{i_v^*} -1 \right),
\]
then the matrix $A + \diag(\frac{1}{i_h^*}, \frac{1}{i_v^*}) B + \diag(-\frac{\beta_h i_v}{i_hi_h^*}, -\frac{C_0\beta_v i_h}{i_vi_v^*})$ is symmetric, and thus Theorem 2.1 in \cite{BERETTA1986677} holds.
\end{proof}

\begin{proof}[Proof of Theorem \ref{fclt-1}] For $t\ge 0$, define
\[
\hatn^n_i(t) = \frac{1}{\sqrt{n}}(N^n_i(nt) -n t).
\]
For $t\ge 0$, we have that
\begin{align*}
\hati^n_h(t) & = \sqrt{n}(\bari^n_h(t)-i_h(t)) \\
& = \hati^n_h(0) + \sqrt{n}\left(\barn^n_1\left(\beta_h^n\int_0^t \bari^n_v(u)\bars_h^n(u) \D u\right) - \beta_h\int_0^t i_v(u)s_h(u)\D u\right)\\
& \quad - \sqrt{n}\left(\barn^n_2\left(\gamma_h^n\int_0^t \bari^n_h(u)\D u\right) - \gamma_h\int_0^t i_h(u)\D u\right) \\
& = \hati^n_h(0) + \hatn^n_1\left(\beta_h^n\int_0^t \bari^n_v(u)\bars_h^n(u) \D u\right) + \sqrt{n} \left(\beta_h^n\int_0^t \bari^n_v(u)\bars_h^n(u) \D u - \beta_h\int_0^t i_v(u)s_h(u)\D u\right) \\
& \quad - \hatn^n_2\left(\gamma_h^n\int_0^t \bari^n_h(u)\D u\right)- \sqrt{n}\left(\gamma_h^n\int_0^t \bari^n_h(u)\D u - \gamma_h\int_0^t i_h(u)\D u\right) \\
& = \hati^n_h(0) + M^n_1(t) + \beta_h\int_0^t \bars^n_h(u) \hati^n_v(u) \D u - \int_0^t (\beta_h i_v(u)+ \gamma_h ) \hati^n_h(u) \D u,
\end{align*}
where
\begin{align*}
M^n_1(t) & =  \hatn^n_1\left(\beta_h^n\int_0^t \bari^n_v(u)\bars_h^n(u) \D u\right) -\hatn^n_2\left(\gamma_h^n\int_0^t \bari^n_h(u)\D u\right) \\
& \quad + \sqrt{n}(\beta^n_h-\beta_h) \int_0^t \bari^n_v(u)\bars^n_h(u)\D u - \sqrt{n}(\gamma^n_h -\gamma_h)\int_0^t \bari^n_h(u) \D u.
\end{align*}
Similarly, we have for $t\ge 0$,
\begin{align*}
\hats^n_v(t)& = \hats^n_v(0) + M^n_2(t)  - \beta_v\int_0^t \bars^n_v(u) \hati^n_h(u) \D u -\beta_v\int_0^t  i_h(u)\hat S^n_v(u)\D u + \gamma_v \int_0^t \hati^n_v(u) \D u, \\
\hati^n_v(t)& = \hati^n_v(0) + M^n_3(t) + \beta_v\int_0^t \bars^n_v(u) \hati^n_h(u) \D u +\beta_v\int_0^t  i_h(u)\hat S^n_v(u)\D u - \gamma_v \int_0^t \hati^n_v(u) \D u,
\end{align*}
where
\begin{align*}
M^n_2(t) & = \hatn^n_3\left(\gamma_v^n\int_0^t \bars^n_v(u) + \bari^n_v(u)\D u\right) -  \hatn^n_4\left(\beta_v^n\int_0^t \bari^n_h(u)\bars_v^n(u) \D u\right) -\hatn^n_5\left(\gamma_v^n\int_0^t \bars^n_v(u)\D u\right) \\
& \quad - \sqrt{n}(\beta^n_v-\beta_v) \int_0^t \bari^n_h(u)\bars^n_v(u)\D u + \sqrt{n}(\gamma^n_v -\gamma_v)\int_0^t \bari^n_v(u) \D u, \\
M^n_3(t) & = \hatn^n_4\left(\beta_v^n\int_0^t \bari^n_h(u)\bars_v^n(u) \D u\right)  - \hatn^n_6\left(\gamma_v^n\int_0^t \bari^n_v(u)\D u\right) \\
& \quad + \sqrt{n}(\beta^n_v-\beta_v) \int_0^t \bari^n_h(u)\bars^n_v(u)\D u - \sqrt{n}(\gamma^n_v -\gamma_v)\int_0^t \bari^n_v(u) \D u.
\end{align*}
We will verify the two conditions in Theorem \ref{tightness} to establish the tightness of $(\hati^n_h, \hats^n_v, \hati^n_v).$ We first note that from the functional central limit theorem for unit-rate Poisson processes, $(\hatn^n_1, \ldots, \hatn^n_6)$ converges to a $6$-dimensional standard Brownian motion. From Theorem \ref{lln-1}, and the random time change theorem, we have that
\begin{align}
(M^n_1, M^n_2, M^n_3)' \Go   (W_1, W_2, W_3)',
\end{align}
where
\begin{align}
W_1(t) & = \int_0^t [\hat\beta_h i_v(u)s_h(u)-\hat\gamma_hi_h(u)]\D u + \int_0^t [\beta_h i_v(u)s_h(u) + \gamma_h i_h(u)] \D B_1(u), \label{BM-111} \\
W_2(t) & = \int_0^t [-\hat\beta_v i_h(u)s_v(u)+\hat\gamma_vi_v(u)]\D u + \gamma_v\int_0^t [C_0 + s_v(u)] \D B_2(u)\nonumber \\
& \quad - \beta_v\int_0^t  i_h(u)s_v(u) \D B_3(u), \label{BM-112}\\
W_3(t) & = \int_0^t [\hat\beta_v i_h(u)s_v(u)-\hat\gamma_vi_v(u)]\D u +  \beta_v \int_0^t  i_h(u)s_v(u) \D B_3(u) +\gamma_v \int_0^t  i_v(u) \D B_4(u),\label{BM-113}
\end{align}
with $B_i, i=1, 2, 3, 4$, independent standard Brownian motions, $i_v$ and $i_u$ defined by \eqref{ode1-1} and \eqref{ode2-1}, $s_u = 1- i_u$, and $s_v = 1- i_v.$
%
In particular $(M_1^n, M^n_2, M^n_3)$ satisfies the conditions inTheorem \ref{tightness} (i) and (ii).
We next observe that for $0\le t_1 \le t_2$,
\begin{align*}
& |(\hati^n_h(t_2), \hats^n_v(t_2), \hati^n_v(t_2)) - (\hati^n_h(t_1), \hats^n_v(t_1), \hati^n_v(t_1))| \\
& \le \sum_{i=1}^3 |M^n_i(t_2) - M^n_i(t_1)|  + \int_{t_1}^{t_2} (\beta_h i_h(u)+\gamma_h + 2\beta_v \bar S^n_v(u))|\hati^n_h(u)| \D u \\
& \quad + \int_{t_1}^{t_2} 2 \beta_v i_h(u) |\hats^n_v(u)| \D u + \int_{t_1}^{t_2} (\beta_h \bar S^n_h(u) + 2\gamma_v) |\hati^n_v(u)| \D u \\
& \le  \sum_{i=1}^3 |M^n_i(t_2) - M^n_i(t_1)|  + \int_{t_1}^{t_2} (2\beta_h+2\beta_v+\gamma_v+2\gamma_v + 2\beta_v \bar S^n_v(u))|(\hati^n_h(u), \hats^n_v(u), \hati^n_v(u)| \D u.
\end{align*}
Using Gronwall's inequality, we have for $T\ge 0$,
\begin{align*}
 \sup_{0\le t \le T} |(\hati^n_h(t), \hats^n_v(t), \hati^n_v(t))| &  \le \left(|(\hati^n_h(0), \hats^n_v(0), \hati^n_v(0))| + \sum_{i=1}^3 \sup_{0\le t \le T}|M^n_i(t)| \right) \\
& \quad \times \exp\left\{\int_0^T(2\beta_h+2\beta_v+\gamma_v+2\gamma_v + 2\beta_v \bar S^n_v(u)) \D u \right\}.
\end{align*}
Noting that $\EE[\sup_{0\le t\le T}|\bar S^n_v(t)-s_v(t)|] \to 0$, and $\sup_{t\ge 0}s_v(t) \le C_0$, we have that for any $M>0,$
\begin{align}\label{relative-compact}
\lim_{M\uparrow\infty}\PP\left(  \sup_{0\le t \le T} |(\hati^n_h(t), \hats^n_v(t), \hati^n_v(t))|> M \right) =0,
\end{align}
which is essentially the condition in Theorem \ref{tightness} (i).
Next for $\delta > 0$,
\begin{align*}
&\sup_{t_1 \le t_2 \le t_1+ \delta}  |(\hati^n_h(t_2), \hats^n_v(t_2), \hati^n_v(t_2)) - (\hati^n_h(t_1), \hats^n_v(t_1), \hati^n_v(t_1))|   \\
& \le \sup_{t_1 \le t_2 \le t_1+ \delta}  \sum_{i=1}^3 |M^n_i(t_2) - M^n_i(t_1)|  \\
& \quad + \sup_{t_1 \le u \le t_1+ \delta} |(\hati^n_h(u), \hats^n_v(u), \hati^n_v(u)|  \int_{t_1}^{t_1+\delta} (2\beta_h+2\beta_v+\gamma_v+2\gamma_v + 2\beta_v \bar S^n_v(u))\D u.
\end{align*}
From \eqref{relative-compact}, we have for any $\epsilon>0,$
\begin{align}\label{jump}
\lim_{\delta\downarrow0}\limsup_{n\to\infty}\PP\left(\sup_{t_1 \le t_2 \le t_1+ \delta}  |(\hati^n_h(t_2), \hats^n_v(t_2), \hati^n_v(t_2)) - (\hati^n_h(t_1), \hats^n_v(t_1), \hati^n_v(t_1))| > \epsilon \right) =0.
\end{align}
The $C$-tightness of $(\hati^n_h, \hats^n_v, \hati^n_v)$ now follows from \eqref{relative-compact} and \eqref{jump}.
Let $(\hati_h, \hats_v, \hati_v)$ be any weak limit, then it satisfies the following integral equations:
\begin{align*}
\hat I_h(t) & = \hat I_h(0) + \int_0^t [\beta_hs_h(u)\hat I_v(u) - (\beta_h i_v(u)+\gamma_h)\hat I_h(u)]\D u + W_1(t), \\
\hat S_v(t) & = \hat S_v(0) + \int_0^t [-\beta_v s_v(u)\hat I_h(u) - \beta_v i_h(u)\hat S_v(u) + \gamma_v \hat I_v(u)]\D u + W_2(t), \\
\hat I_v(t) & = \hat I_v(0) + \int_0^t [\beta_v s_v(u)\hat I_h(u) + \beta_v i_h(u)\hat S_v(u) - \gamma_v \hat I_v(u)]\D u + W_3(t),
 \end{align*}
 where $W_1, W_2,$ and $W_3$ are defined in \eqref{BM-111}--\eqref{BM-113}.
Noting that $s_h, i_h, s_v, i_v$ are all uniformly bounded, from \cite{Oksendal03}, there exists a unique solution to the above integral equations. The theorem follows.
\end{proof}

\subsection{Proofs for Case II}
To prove the theorems (Theorems \ref{lln-2}, \ref{lln-stability-II}, and \ref{fclt-2}) for Case II, we study the generators of the Markov processes $(\bar I^n_h, \bar S^n_v, \bar I^n_v).$
Recall that
\begin{align*}
\bar I^n_h(t) & = \bar I^n_h(0) + \bar N^n_1\left( \beta^n_h\int_0^t \bar I^n_v(u)\bar S^n_h(u)\D u\right) - \bar N^n_2\left(\gamma^n_h\int_0^t \bar I^n_h(u) \D u \right), \\
\bar S^n_v(t) & = \bar S^n_v(0) + \bar N^n_3\left(\gamma^n_v\int_0^t \bar S^n_v(u) + \bar I^n_v(u) \D u\right) - \bar N^n_4\left(\beta^n_v\int_0^t \bar I^n_h(u) \bar S^n_v(u) \D u\right)\\
& \quad  - \bar N^n_5\left(\gamma^n_v\int_0^t \bar S^n_v(u) \D u\right), \\
\bar I^n_v(t) & = \bar I^n_v(0) + \bar N^n_4\left(\beta^n_v\int_0^t \bar I^n_h(u) \bar S^n_v(u) \D u\right) - \bar N^n_6\left(\gamma^n_v\int_0^t \bar I^n_v(u) \D u\right).
\end{align*}
Then for $f\in \mathcal{C}^2_0(\RR^3_+)$, we have
\begin{align*}
& f(\bar I^n_h(t), \bar S^n_v(t), \bar I^n_v(t))  = f(\bar I^n_h(0), \bar S^n_v(0), \bar I^n_v(0)) \\
&  + \int_0^t \left[f(\bar I^n_h(u-)+\frac{1}{n}, \bar S^n_v(u-), \bar I^n_v(u-)) - f(\bar I^n_h(u-), \bar S^n_v(u-), \bar I^n_v(u-)) \right] \D N^n_1\left(n\beta_h^n \int_0^u \bar I_v^n(\tau) \bar S_h^n(\tau) \D \tau\right) \\
& + \int_0^t \left[f(\bar I^n_h(u-)-\frac{1}{n}, \bar S^n_v(u-), \bar I^n_v(u-)) - f(\bar I^n_h(u-), \bar S^n_v(u-), \bar I^n_v(u-)) \right] \D N_2^n\left(n \gamma_h^n \int_0^u \bar I_h^n(\tau) \D \tau \right)\\
&  + \int_0^t \left[f(\bar I^n_h(u-), \bar S^n_v(u-)+\frac{1}{n}, \bar I^n_v(u-)) - f(\bar I^n_h(u-), \bar S^n_v(u-), \bar I^n_v(u-)) \right] \D N^n_3\left(n \gamma^n_v\int_0^u [\bar S^n_v(\tau) + \bar I^n_v(\tau)] \D \tau\right) \\
& + \int_0^t \left[f(\bar I^n_h(u-), \bar S^n_v(u-)-\frac{1}{n}, \bar I^n_v(u-)) - f(\bar I^n_h(u-), \bar S^n_v(u-), \bar I^n_v(u-)) \right] \D N^n_4\left(n \beta^n_v\int_0^u \bar I^n_h(\tau) \bar S^n_v(\tau) \D \tau\right)\\
& + \int_0^t \left[f(\bar I^n_h(u-), \bar S^n_v(u-)-\frac{1}{n}, \bar I^n_v(u-)) - f(\bar I^n_h(u-), \bar S^n_v(u-), \bar I^n_v(u-)) \right] \D N^n_5\left(n\gamma^n_v\int_0^u \bar S^n_v(\tau) \D \tau\right)\\
&  +\int_0^t \left[f(\bar I^n_h(u-), \bar I^n_v(u-)+\frac{1}{n}) - f(\bar I^n_h(u-), \bar I^n_v(u-)) \right] \D N_4^n\left(n \beta_v^n  \int_0^{ u} \bar I_h^n(u) \bar S_v^n(\tau) \D \tau \right)\\
& + \int_0^t \left[f(\bar I^n_h(u-), \bar I^n_v(u-)-\frac{1}{n}) - f(\bar I^n_h(u-), \bar I^n_v(u-)) \right] \D N_6^n \left(n \gamma_v^n \int_0^u \bar I_v^n(\tau) \D \tau \right).
\end{align*}
Define the following generators $\cla^n_1, \cla^n_2, \cla^n_3$ and $\cla^n, \clb^n$ for $f\in \mathcal{C}^2_0(\RR^3_+)$, and $(x_1, x_2, x_3)\in \RR^3_+$.
\begin{align*}
\cla^n_1f(x_1, x_2, x_3) & = n\beta^n_h x_3(1-x_1) \left[f(x_1+\frac{1}{n}, x_2, x_3) - f(x_1,x_2, x_3) \right] \\
& \quad + n\gamma^n_h x_1  \left[f(x_1-\frac{1}{n}, x_2, x_3) - f(x_1,x_2, x_3) \right] \\
& \quad + n\gamma^n_v (x_2+x_3)  \left[f(x_1, x_2+\frac{1}{n}, x_3) - f(x_1,x_2, x_3) \right] \\
& \quad + n\beta_v^n x_1x_2  \left[f(x_1, x_2-\frac{1}{n}, x_3) - f(x_1,x_2, x_3) \right] \\
& \quad + n\gamma_v^n x_2  \left[f(x_1, x_2-\frac{1}{n}, x_3) - f(x_1,x_2, x_3) \right] \\
& \quad + n\beta^n_v x_1x_2  \left[f(x_1, x_2 x_3+\frac{1}{n}) - f(x_1,x_2,x_3) \right]\\
& \quad + n\gamma^n_v x_3 \left[f(x_1, x_2, x_3-\frac{1}{n}) - f(x_1,x_2,x_3) \right],\\
\cla^n_2 f(x_1, x_2, x_3) & = n\beta^n_h x_3(1-x_1) \left[f(x_1+\frac{1}{n}, x_2, x_3) - f(x_1,x_2, x_3) \right] \\
& \quad + n\gamma^n_h x_1  \left[f(x_1-\frac{1}{n}, x_2, x_3) - f(x_1,x_2, x_3) \right], \\
\mathcal{A}^n_3 f(x_1,x_2, x_3) & =  n\gamma^n_v (x_2+x_3)  \left[f(x_1, x_2+\frac{1}{n}, x_3) - f(x_1,x_2, x_3) \right] \\
& \quad + n\beta_v^n x_1x_2  \left[f(x_1, x_2-\frac{1}{n}, x_3) - f(x_1,x_2, x_3) \right] \\
& \quad + n\gamma_v^n x_2  \left[f(x_1, x_2-\frac{1}{n}, x_3) - f(x_1,x_2, x_3) \right] \\
& \quad + n\beta^n_v x_1x_2  \left[f(x_1, x_2, x_3+\frac{1}{n}) - f(x_1,x_2,x_3) \right]\\
& \quad + n\gamma^n_v x_3 \left[f(x_1, x_2, x_3-\frac{1}{n}) - f(x_1,x_2,x_3) \right],
\end{align*}
and
\begin{align}
\cla^n f(x_1, x_2,x_3) & =[\beta_h^n x_3(1-x_1) - \gamma_h^n x_1]\frac{\partial f}{\partial x_1}(x_1, x_2, x_3), \label{generator-1}\\
\clb^n f(x_1, x_2,x_3) & =\frac{1}{\alpha(n)}[\gamma_v^n x_3-\beta_v^n x_1x_2]  \frac{\partial f}{\partial x_2}(x_1, x_2,x_3) \\
& \quad +\frac{1}{\alpha(n)}[\beta_v^n x_1x_2-\gamma_v^n x_3]  \frac{\partial f}{\partial x_3}(x_1, x_2,x_3). \label{generator-2}
\end{align}
Then from Taylor expansion, we have for $f\in C^2_0(\mathbb{R}^3_+)$,
\begin{align}
\cla^n_1 f(x_1, x_2, x_3) &= \cla^n f(x_1, x_2, x_3) + \alpha(n) \clb^n f(x_1, x_2, x_3) + O(n^{-1}) + O(\alpha(n)/n). \label{generator-4} \\
\cla^n_2 f(x_1, x_2, x_3) &= \cla^n f(x_1, x_2, x_3) + O(n^{-1}), \label{generator-3} \\
\cla^n_3 f(x_1, x_2, x_3) & =\alpha(n) \clb^n f(x_1, x_2, x_3) + O(\alpha(n)/n). \label{generator-55}
\end{align}

\begin{proof}[Proof of Theorems \ref{lln-2} and \ref{lln-stability-II}] The main idea is adapted from Theorem 2.1 in \cite{10.1007/BFb0007058}. We first define the following $\sigma$-field: For $t\ge 0,$
\[
\mathcal{F}_t^n = \sigma(N^n_i(s), i=1, \ldots, 6, I^n_h(s), S^n_v(s), I^n_v(s); s\le t).
\]
From \eqref{meanvecsize}, we have that for $t\ge 0,$
\begin{align*}
\EE[\bars^n_v(t)+\bari^n_v(t)] = C_0.
\end{align*}
We next observe that $\bars^n_v + \bari^n_v$ is a $\{\mathcal{F}_t\}$ martingale, and
\[
\EE[\langle \bars^n_v + \bari^n_v \rangle_t] =  \frac{2\gamma^n_v}{n}\int_0^t \EE[\bars^n_v(u) + \bari^n_v(u)] \D u \to 0.
\]
From the martingale central limit theorem, we have that
\begin{align}\label{vector-size-conv}
\bars^n_v + \bari^n_v \Go C_0.
\end{align}
We next establish the relative compactness of $\Gamma^n$. From \eqref{meanvecsize} and the Markov inequality, we have for $t\ge 0$ and $L>0$,
\begin{align*}
\EE[\Gamma^n([0, t]\times [0, L]\times [0, L]) ] & = \int_0^t \PP(\bars^n_v(s) \le L, \bari^n_v(s) \le L) \D s \\
& \ge \int_0^t \PP(\bars^n_v(s) + \bari^n_v(s) \le L) \D s \\
& = t - \int_0^t \PP(\bars^n_v(s) + \bari^n_v(s) > L) \D s \\
& \ge t \left(1- \frac{C_0}{L}\right),
\end{align*}
which implies that for each $t\ge 0$ and $\epsilon>0$, there exists a compact set $K\in\mathcal{B}(\RR^2_+)$ such that
\[
\inf_n\EE[\Gamma^n([0,t]\times K)] \ge (1-\epsilon) t.
\]
Then, from Lemma 1.3 in \cite{10.1007/BFb0007058}, $\{\Gamma^n\}_{n\ge 1}$ is relatively compact. We next show the tightness of $\{\bar I^n_h\}$. For $t\ge 0$, define
\begin{align}\label{martingale-111}
M^n(t) = \bari^n_h(t) -\bari^n_h(0)- \beta^n_h\int_0^t \bari^n_v(u)\bars^n_h(u) \D u+ \gamma^n_h\int_0^t\bari^n_h(u) \D u.
\end{align}
Then $M^n$ is a martingale, and
\[
\langle M^n \rangle_t  = \frac{\beta^n_h}{n} \int_0^t \bari^n_v(u)\bars^n_h(u) \D u+\frac{\gamma^n_h}{n} \int_0^t\bari^n_h(u) \D u \to 0,
\]
which yields that $M^n\Go 0.$ From \eqref{martingale-111}, we have
\[
\bari^n_h(t) =  \bari^n_h(0)+\int_0^t [\beta^n_h \bari^n_v(u)\bars^n_h(u) -  \gamma^n_h\bari^n_h(u)] \D u - M^n(t).
\]
We note that there exists a positive constant $c_1$ such that
\begin{align*}
\int_0^t\EE|\beta^n_h \bari^n_v(u)\bars^n_h(u) -  \gamma^n_h\bari^n_h(u)| \D u \le c_1 t.
\end{align*}
Then both conditions in Theorem \ref{tightness} can be verified easily, and it follows that $\bari^n_h$ is $C$-tight. Thus $\{(\bar I^n_h, \Gamma^n)\}$ is relatively compact. Now, let $(\bar I_h, \bar\Gamma)$ be a limit point along a subsequence $\{n'\}$. Note that for $f\in C^2_0(\RR_+^2)$,
\begin{align*}
&f(\bars^n_v(t), \bar I^n_v(t)) - f(\bars^n_v(0), \bar I^n_v(0)) - \int_0^t \cla_3^n f(\bar I^n_h(u), \bars^n_v(u), \bar I_v^n(u)) \D u \\
& = f(\bars^n_v(t), \bar I^n_v(t)) - f(\bars^n_v(0), \bar I^n_v(0)) - \int_0^t \int_{\RR^2_+} \cla_3^n f(\bar I^n_h(u), y) \Gamma^n(\D u\times \D y ),
\end{align*}
which is a martingale. Dividing it by $\alpha(n')$ and letting $n'\to\infty$, we have
\begin{align}\label{martingale-3}
\int_{[0,t]}\int_{\RR^2_+} \clb f(\bar I_h(u), y) \bar\Gamma(\D u \times \D y)
\end{align}
is a local martingale, where for $f\in \mathcal{C}_0^2(\RR_+^2)$,
\begin{align}\label{generator-5}
\clb f(x_1, x_2, x_3) = [\gamma_v x_3 - \beta_v x_1x_2]\frac{\partial f}{\partial x_2}(x_2, x_3) +[\beta_v x_1x_2-\gamma_v x_3] \frac{\partial f}{\partial x_3}(x_2, x_3).
\end{align}
  From Lemma 1.4 in \cite{10.1007/BFb0007058}, there exists a $\pp(\RR_+^2)$-valued process $\gamma$ such that for any measurable function $l: [0,\infty)\times \RR_+^2 \to \RR$,
\[
\int_{[0,t]}\int_{\RR^2_+} l(s,y) \bar\Gamma(\D s\times \D y) = \int_0^t \int_{\RR^2_+} l(s,y) \gamma_s(\D y) \D s, \ t\ge 0.
\]
Thus from \eqref{martingale-3} for all $t\ge 0$, and $f\in C^2_0(\RR^2_+)$, we have
\begin{align}\label{measure}
\int_{[0,t]}\int_{\RR^2_+} \clb f(\bar I_h(u), y) \bar\Gamma(\D u \times \D y)=\int_0^t \int_{\RR^2_+} \clb f(\bar I_h(u), y) \gamma_u(\D y) \D u = 0,
\end{align}
where the last equality follows from the fact that \eqref{martingale-3} is a the local martingale that is continuous and of bounded variation.
Consequently, we have for a.e. $t$,
\[
\int_{\RR^2_+} \clb f(\bar I_h(t), y) \gamma_t(\D y) =0.
\]
Then, from \eqref{generator-5}, we have for a.e. $t$,
\begin{align}
0 & = \int_{\RR^2_+} \clb f(\bar I_h(t), x) \gamma_t(\D x) \nonumber\\
  &  = \int_{\RR^2_+} [\gamma_v x_3 - \beta_v \bari^n_h(t)x_2]\frac{\partial f}{\partial x_2}(x_2, x_3) +[\beta_v \bari^n_h(t)x_2-\gamma_v x_3] \frac{\partial f}{\partial x_3}(x_2, x_3) \gamma_t(\D x_2\times \D x_3).\label{eqn:4.2}
\end{align}
For a given $x_1\in\RR_+$, let
\[
f(x_2, x_3) =  \gamma_v \frac{x^2_3}{2} -\beta_v x_1\frac{x_2^2}{2}.
\]
Putting the above $f$ into \eqref{eqn:4.2}, we have for any $f\in \mathcal{C}_0^2(\RR_+)$ and a.e. $t$,
\begin{align}\label{vector-conv}
\int_{\RR^2_+} [\beta_v \bari_h(t)x_2-\gamma_v x_3]^2 \gamma_t(\D x_2\times \D x_3) = 0.
\end{align}
Combining \eqref{vector-size-conv} and \eqref{vector-conv}, we have for a.e. $t$,
\begin{align}\label{average}
\gamma_t (\D x_2\times \D x_3) = \delta_{\left(\frac{C_0\gamma_v}{\beta_v\bari_h(t)+\gamma_v}, \ \frac{C_0\beta_v\bari^n_h(t)}{\beta_v\bari_h(t)+\gamma_v}\right)}(\D x_2\times \D x_3).
\end{align}
We next note that for $g \in C^2_0(\RR_+)$,
\begin{align}\label{martingale-2}
g(\bar I^n_h(t)) - g(\bar I^n_h(0)) - \int_{[0,t]}\int_{\RR^2_+} \cla^n g(\bar I^n_h(u), y) \Gamma^n(\D u \times \D x)
\end{align}
is a martingale. Letting $n'\to\infty$, we have that
\begin{align}\label{martingale-2}
g(\bar I_h(t)) - g(\bar I_h(0)) - \int_{[0,t]}\int_{\RR^2_+} \cla g(\bar I_h(u), y) \bar\Gamma(\D u \times \D x)
\end{align}
is a local martingale, where for $(x_1,x_2, x_3)\in \RR^2_+$,
\begin{align*}
\cla g(x_1,x_2, x_3) & =\lim_{n\to\infty}\cla^n g(x_1, x_2, x_3) = [\beta_hx_3 (1-x_1)-\gamma_hx_1]g'(x_1).
\end{align*}
From \eqref{average}, we have that for $g\in C^2_c(\RR_+),$
\begin{align*}
g(\bar I_h(t)) - g(\bar I_h(0)) - \int_0^t \cla g\left(\bar I_u(u),  \frac{C_0\gamma_v }{\beta_v \bar I_h(u) +\gamma_v}, \frac{C_0\beta_v \bar I_h(u)}{\beta_v \bar I_h(u) +\gamma_v}\right) \D u
\end{align*}
is a local martingale. Let $g$ be the identity function, then
\begin{align*}
\bar I_h(t) - \bar I_h(0) - \int_0^t \left( \frac{C_0\beta_h\beta_v \bar I_h(u)(1-\bar I_h(u))}{\beta_v \bar I_h(u) +\gamma_v} - \gamma_h \bar I_h(u) \right)\D u
\end{align*}
is a martingale. Finally, noting that the above maringale is continuous and of bounded variation, it must be $0$. Thus we have shown that
\begin{align*}
\bar I_h(t) = \bar I_h(0) + \int_0^t \left( \frac{C_0\beta_h\beta_v \bar I_h(u)(1-\bar I_h(u))}{\beta_v \bar I_h(u) +\gamma_v} - \gamma_h \bar I_h(u) \right)\D u.
\end{align*}
It can be verified, in the same way as in the proof of Theorem \ref{lln-1}, that the above system has a unique solution. Define
\[
h(x) = \frac{C_0 \beta_h \beta_v x (1-x)}{\beta_v x +\gamma_v} - \gamma_h x, \ x\in [0,1].
 \]
The locally asymptotic stability can be shown using linearization. For the global and exponential asymtotic stability, we can use the Lyapunov function $L(i_h) = i_h$ for $E^e$ and $L(i_h) = i_h-i_h^* - i_h^* \log(i_h/i_h^*)$ for $E^f$.
\end{proof}

\begin{proof}[Proof of Theorem \ref{fclt-2}]
Let $f(x_1, x_2) = x_1$ and $F^n(x_1, x_3) = \beta_h^n x_3(1-x_1) - \gamma_h^n x_1$. We observe that
\begin{align*}
M^n_1(t) := \bar I^n_h(t) - \bar I^n_h(0)-  \int_0^t F^n(\bar I^n_h(u), \bar I^n_v(u))\D u
\end{align*}
is a martingale. Therefore,
\begin{align*}
\hat I^n_h(t) := \sqrt{n} (\bar I^n_h(t) - i_h(t)) = \sqrt{n} \bar I^n_h(0) +\sqrt{n}  M^n_1(t)  + \sqrt{n}  \int_0^t F^n(\bar I^n_h(u), \bar I^n_v(u))\D u - \sqrt{n}  i_h(t).
\end{align*}
 Let
\begin{align*}
F(x_1, x_3) & = \lim_{n\to\infty} F^n(x_1, x_3) = \beta_h x_3(1-x_1) - \gamma_h x_1, \\
F^{*}(x_1)& = F(x_1, C_0\beta_v x_1/(\beta_v x_1 + \gamma_v)), \\
\bar F(x_1, x_3) & = F(x_1, x_3) - F^{*}(x_1)  = \beta_h(1- x_1) \left[x_3- \frac{C_0\beta_v x_1}{\beta_v x_1 + \gamma_v}\right].
\end{align*}
Define
\[
h(x_1, x_3) = -\frac{\beta_h x_3(1-x_1)}{\beta_v x_1 +\gamma_v},
\]
and let $\tilde h^n = h/\alpha(n).$ Then
\begin{align*}
M^n_2(t)&: =  \tilde h^n(\bar I^n_h(t), \bar I^n_v(t)) - \tilde h^n(\bar I^n_h(0), \bar I^n_v(0)) -\int_0^t \cla^n_1 \tilde h^n(\bar I^n_h(u), \bar I^n_v(u)) \D u
\end{align*}
is a martingale. Thus
\begin{align*}
& \int_0^t F(\bar I^n_h(u), \bar I^n_v(u))\D u \\
& = \int_0^t F^{*}(\bar I^n_h(u))\D u+ \int_0^t \bar F(\bar I^n_h(u), \bar I^n_v(u)) \D u \\
& = \int_0^t F^{*}(\bar I^n_h(u))\D u+ \int_0^t \cla_1^n  \tilde h^n(\bar I^n_h(u), \bar I^n_v(u))\D u + \int_0^t [\bar F(\bar I^n_h(u), \bar I^n_v(u)) - \cla_1^n  \tilde h^n(\bar I^n_h(u), \bar I^n_v(u))]\D u\\
& = \int_0^t F^{*}(\bar I^n_h(u))\D u - M^n_2(t) + [\tilde h^n(\bar I^n_h(t), \bar I^n_v(t)) - \tilde h^n(\bar I^n_h(0), \bar I^n_v(0))] \\
& \quad + \int_0^t [\bar F(\bar I^n_h(u), \bar I^n_v(u)) - \cla_1^n \tilde h^n(\bar I^n_h(u), \bar I^n_v(u))]\D u.
\end{align*}
From \eqref{ode-2}, $i_h(t) = \int_0^t F^*(i_h(u)) \D u.$
Hence we have
\begin{align*}
\hat I^n_h(t) 
& = \hat I^n_h(0) +\sqrt{n} M^n_1(t)-  \sqrt{n} M^n_2(t) + \int_0^t \sqrt{n}[F^{n}(\bar I^n_h(u), \bari^n_v(u))-F(\bari_h^n(u), \bari_v^n(u))] \D u \\
& \quad + \int_0^t \sqrt{n}[F^{*}(\bar I^n_h(u))-F^*(i_h(u))] \D u   + \sqrt{n} [\tilde h^n(\bar I^n_h(t), \bar I^n_v(t)) - \tilde h^n(\bar I^n_h(0), \bar I^n_v(0))] \\
& \quad  + \sqrt{n}\int_0^t  [\bar F(\bar I^n_h(u), \bar I^n_v(u)) -\cla_1^n \tilde  h^n(\bar I^n_h(u), \bar I^n_v(u))]\D u.
\end{align*}
Let
\[
\hat M^n_i(t) = \sqrt{n}M^n_i(t), i=1,2.
\]
We next observe that for $t\ge 0$,
\begin{align*}
[\hat M^n_1, \hat M^n_1]_t = \frac{1}{n}\left[ N^n_1\left(n\beta_h^n \int_0^t \bar I_v^n(u) \bar S_h^n(u) \D u\right) + N_2^n\left(n \gamma_h^n \int_0^t  \bar I_h^n(u) \D u \right)\right].
\end{align*}
Therefore, we have from similar arguments made in the proof of Theorem \ref{lln-2},
\[
[\hat M^n_1, \hat M^n_1]_t \to \beta_h\int_0^t i^*_v(u)(1- i_h(u)) \D u + \gamma_h \int_0^t i_h(u)\D u,
\]
where $i_v^*(u) = \frac{C_0\beta_v i_h(u)}{\beta_v i_h(u) +\gamma_v}$ for $u\ge 0,$
and so
\begin{align}
\hat M^n_1 \Go\int_0^t  \sqrt{\beta_hi^*_v(u)(1- i_h(u)) + \gamma_h i_h(u) }\ \D B(u),
\end{align}
where $B$ is a standard Brownian motion.
Now consider $\hat M^n_2.$ We have
\begin{align*}
& [\hat M^n_2, \hat M^n_2]_t  \\
& =n \int_0^t \left[\tilde h^n(\bar I^n_h(u-)+\frac{1}{n}, \bar I^n_v(u-)) - \tilde h^n(\bar I^n_h(u-), \bar I^n_v(u-)) \right]^2 \D N^n_1\left(n\beta_h^n \int_0^u \bar I_v^n(v) \bar S_h^n(v) \D v\right) \\
&\quad + n\int_0^t \left[\tilde h^n(\bar I^n_h(u-)-\frac{1}{n}, \bar I^n_v(u-)) - \tilde h^n(\bar I^n_h(u-), \bar I^n_v(u-)) \right]^2 \D N_2^n\left(n \gamma_h^n \int_0^u  \bar I_h^n(v) \D v \right)\\
& \quad +\frac{1}{n \alpha(n)^2}\int_0^t \left[\frac{\beta_h (1-\bar I^n_h(u))}{\beta_v x_1 +\gamma_v}\right]^2 \D N_4^n\left(n \beta_v^n  \int_0^{ u} \bar I_h^n(v) \bar S_v^n(v) \D v \right)\\
&\quad + \frac{1}{n\alpha(n)^2}\int_0^t \left[\frac{\beta_h (1-\bar I^n_h(u))}{\beta_v x_1 +\gamma_v}\right]^2 \D N_6^n \left(n \gamma_v^n \int_0^u \bar I_v^n(v) \D v \right) \\
& \to 0,
\end{align*}
which implies that
\begin{align}
\hat M^n_2 \Go 0.
\end{align}
Next, using Theorem \ref{lln-2},  we observe that
\begin{align}
\int_0^t \sqrt{n}[F^{n}(\bar I^n_h(u), \bari^n_v(u))-F(\bari_h^n(u), \bari_v^n(u))] \D u \to \int_0^t \hat\beta_h i_v^*(u)(1-i_h(u)) - \hat\gamma_h i_h(u) \D u,
\end{align}
and
\begin{align}
\sqrt{n} [\tilde h^n(\bar I^n_h(t), \bar I^n_v(t)) - \tilde h^n(\bar I^n_h(0), \bar I^n_v(0))]  \to 0, 
\end{align}
and
\begin{align}
& \sqrt{n}\int_0^t  [\bar F(\bar I^n_h(u), \bar I^n_v(u)) -\cla_1^n \tilde  h^n(\bar I^n_h(u), \bar I^n_v(u))]\D u \nonumber\\
& = \sqrt{n}\int_0^t  [\bar F(\bar I^n_h(u), \bar I^n_v(u)) -\cla^n \tilde h^n(\bar I^n_h(u), \bar I^n_v(u)) - \clb^n h(\bar I^n_h(u), \bar I^n_v(u)) - O(n^{-1})]\D u \nonumber\\
& = O(1/\sqrt{n}))  - \frac{\sqrt{n}}{\alpha(n)}\int_0^t \cla^n h(\bar I^n_h(u), \bar I^n_v(u)) \D u \nonumber\\
& \quad + \sqrt{n} \left[\int_0^t \beta_h(1- \bari^n_h(u)) \left[\bari^n_v(u)- \frac{C_0\beta_v \bari^n_h(u)}{\beta_v \bari^n_h(u) + \gamma_v}\right] \right. \nonumber\\
&\quad   +\left.\beta_h(1- \bari^n_h(u))  \frac{C_0\beta_v^n \bari^n_h(u)- (\beta_v^n \bari^n_h(u) +\gamma_v^n)\bari^n_v(u)}{{\alpha(n)}(\beta_v \bari^n_h(u) +\gamma_v)} \D u \right] \nonumber\\
& \to 
 \int_0^t \frac{\beta_h(1- \bari^n_h(u))(C_0 \hat\beta_v i_h(u) - \hat\beta_vi_h(u)i^*_v(u) - \hat\gamma_v i^*_v(u))}{\beta_v i_h(u)+\gamma_v} \D u.
\end{align}
At last, we note that
\begin{align*}
& \int_0^t \sqrt{n}[F^{*}(\bar I^n_h(u))-F^*(i_h(u))] \D u  \\
 =& \int_0^t \frac{G\gamma_v\hati^n_h(u) - H\beta_v\bari^n_h(u)i_h(u)\hati^n_h(u)-H\gamma_v (\bari^n_h(u)+i_h(u))\hati^n_h(u)}{(\beta_v \hati^n_h(u)+\gamma_v)(\beta_v i_h(u)+\gamma_v)} \D u,
\end{align*}
where $G = C_0\beta_h\beta_v - \gamma_h\gamma_v$ and $H = C_0\beta_h\beta_v + \gamma_h\beta_v.$
We can show the $C$-tightness of $\hati^n_h$ following similar steps to those in the proof of Theorem \ref{fclt-1}. Let $\hati_h$ be a weak limit of $\hati^n_h$. Then
\begin{align*}
\hati_h(t) = \hati_h(0) + \int_0^t  \sqrt{\beta_hi^*_v(u)(1- i_h(u)) + \gamma_h i_h(u)} \ \D B(u) + \int_0^t D(i_h(u), \hati_h(u)) \D u,
\end{align*}
where
\begin{align*}
D(i_h(u), \hati_h(u)) & =\hat\beta_h i_v^*(u)(1-i_h(u)) - \hat\gamma_h i_h(u)\\
& \quad + \frac{\beta_h(1- i_h(u))(C_0 \hat\beta_v i_h(u) - \hat\beta_vi_h(u)i^*_v(u) - \hat\gamma_v i^*_v(u))}{\beta_v i_h(u)+\gamma_v}  \\
& \quad +\frac{(G\gamma_v - H\beta_v i_h(u)^2-2H\gamma_v i_h(u))\hati_h(u)}{(\beta_v i_h(u)+\gamma_v)^2}.
\end{align*}

\end{proof}

\subsection{Proof of Lemma \ref{infinite_var}}

We study the eigenvalues and eigenvectors of $C_e$. The right eigenvector corresponding to the eigenvalue $0$ can be taken to be 
\[
\mathbf{v}_{e,0} = (u_1, u_2, u_3)',
\]
where $u_1, u_2, u_3$ are positive and satisfy the linear equations $-(\gamma_h + \beta_h i^e_v) u_1 + \beta_h s^e_h u_3 =0$ and $-\beta_v s^e_v u_1 -\beta_v i^e_h u_2 + \gamma_v u_3 =0$. We next note that the two negative eigenvalues are the same as the eigenvalues of \eqref{fixed-popl}. Denote by $(v_{11}, v_{21})'$ and $(v_{12}, v_{22})'$ the right eigenvectors corresponding to the eigenvalues $\lambda_{e,1}$ and $\lambda_{e,2}$ for \eqref{fixed-popl}. It then follows that $(v_{11}, -v_{21}, v_{21})'$ and $(v_{12}, -v_{22}, v_{22})'$ are the right eigenvectors corresponding to the eigenvalues $\lambda_{e,1}$ and $\lambda_{e,2}$ for $C_e$ in \eqref{rand-popl}. From \eqref{me-app-2}, we have 
\begin{align*}
e^{C_et}= \begin{pmatrix} u_1 & v_{11} & v_{12} \\ u_2 & -v_{21}& -v_{22}\\ u_3 & v_{21} & v_{22}\end{pmatrix} \begin{pmatrix} 1 & 0 & 0 \\ 0 & e^{\lambda_{e,1} t}& 0\\ 0 & 0 & e^{\lambda_{e,2} t} \end{pmatrix}\begin{pmatrix} u_1 & v_{11} & v_{12} \\ u_2 & -v_{21}& -v_{22}\\ u_3 & v_{21} & v_{22}\end{pmatrix}^{-1}. 
\end{align*}
We next write 
\[
\begin{pmatrix} v_{11} & v_{12} \\  v_{21} & v_{22}\end{pmatrix}^{-1} = \begin{pmatrix} a_{11} & a_{12} \\  a_{21} & a_{22}\end{pmatrix}.
\]
It follows that 
\[
\begin{pmatrix} u_1 & v_{11} & v_{12} \\ u_2 & -v_{21}& -v_{22}\\ u_3 & v_{21} & v_{22}\end{pmatrix}^{-1} = \begin{pmatrix} 0 & \frac{1}{u_2+u_3} &  \frac{1}{u_2+u_3} \\ a_{11}&  -\frac{a_{11}u_1+a_{12}u_3}{u_2+u_3} & -\frac{a_{11}u_1-a_{12}u_2}{u_2+u_3}\\ a_{21} & -\frac{a_{21}u_1+a_{22}u_3}{u_2+u_3}& -\frac{a_{21}u_1-a_{22}u_2}{u_2+u_3}\end{pmatrix},
\]
and $e^{C_e t} =$
\begin{align*}
& \begin{pmatrix} 0 & u_1/(u_2+u_3) & u_1/(u_2+u_3) \\ 0  & u_2/(u_2+u_3)& u_2/(u_2+u_3)\\ 0 &  u_3/(u_2+u_3) & u_3/(u_2+u_3)\end{pmatrix} \\
& + \begin{pmatrix} v_{11}a_{11}e^{\lambda_{e,1}t} + v_{12}a_{21}e^{\lambda_{e,2}t} &  v_{11}b_1 e^{\lambda_{e, 1}t} + v_{12} b_2 e^{\lambda_{e,2} t}& v_{11}b_3 e^{\lambda_{e, 1}t} + v_{12} b_4 e^{\lambda_{e,2} t}\\ -v_{21}a_{11}e^{\lambda_{e,1}t} - v_{22}a_{21}e^{\lambda_{e,2}t}  & -v_{21}b_1 e^{\lambda_{e, 1}t} - v_{22} b_2 e^{\lambda_{e,2} t}& -v_{21}b_3 e^{\lambda_{e, 1}t} - v_{22} b_4 e^{\lambda_{e,2} t}\\ v_{21}a_{11}e^{\lambda_{e,1}t} + v_{22}a_{21}e^{\lambda_{e,2}t} &  v_{21}b_1 e^{\lambda_{e, 1}t} + v_{22} b_2 e^{\lambda_{e,2} t} & v_{21}b_3 e^{\lambda_{e, 1}t} + v_{22} b_4 e^{\lambda_{e,2} t}\end{pmatrix},
\end{align*}
where $b_1 = - (u_1a_{11}+a_{12}u_3)/(u_2+u_3), b_2 = -(u_1a_{21}+u_3a_{22})/(u_2+u_3), b_3 = -(u_1a_{11}-u_2a_{12})/(u_2+u_3)$, and $b_4 = -(u_1a_{21}-u_2a_{22})/(u_2+u_3).$
It follows from \eqref{cov-app-1} that the variances of $\hati_h(t), \hats_v(t),$ and $\hati_v(t)$ approach $\infty$ as $t\to\infty.$

\section{Appendix: Notation}\label{table}
\begin{table}[H]
   \centering
   \begin{tabular}{cp{4.3in}}
 Parameter		&	Definition \\
 \hline
  $n$		&	Size of host population.  \vspace{0.05in}\\
 $C_0$	& Multiple by which vector population scale to host population.\\
  \hline\\
  $\beta_h^n$ & Disease transmission rate to hosts from a typical infecious vector.\vspace{0.05in} \\
  $\gamma_h^n$ & Recovery rate of a typical infected host. \vspace{0.05in}\\
   $\beta_h$ &  $\lim_{n\to\infty}\beta^n_h$. \vspace{0.05in}\\
 $\gamma_h$ & $\lim_{n\to\infty}\gamma^n_h$. \vspace{0.05in}\\
 $\hat\beta_h$ & $\lim_{n\to\infty}\sqrt{n}(\beta^n_h-\beta_h)$. \vspace{0.05in}\\
$\hat\gamma_h$ & $\lim_{n\to\infty}\sqrt{n}(\gamma^n_h-\gamma_h)$.\vspace{0.05in}\\
\hline \\
$\beta_v^n$ & Disease transmission rate to vectors from a typical infecious host.\vspace{0.05in}\\
$\gamma_v^n$ & Equal birth and death rate of a vector.\vspace{0.05in}\\
Parameter scale of Case I: & Both hosts and vectors evolve at rate $\mathcal{O}(1)$ as $n\to\infty$.   \vspace{0.05in}\\
    $\beta_v $ & $\lim_{n\to\infty}\beta^n_v$. \vspace{0.05in}\\
$\gamma_v $ & $\lim_{n\to\infty}\gamma^n_v$. \vspace{0.05in}\\
 $\hat\beta_v$ & $ \lim_{n\to\infty}\sqrt{n}(\beta^n_v-\beta_v)$. \vspace{0.05in}\\
    $\hat\gamma_v$ & $\lim_{n\to\infty}\sqrt{n}(\gamma^n_v-\gamma_v)$. \vspace{0.1in}\\
Parameter scale of Case II:& Vectors evolve much faster at rate $\mathcal{O}(\alpha(n))$ and hosts evolve at rate $\mathcal{O}(1)$ as $n\to\infty$, where $\alpha(n)/n\to 0$ and $\alpha(n)/\sqrt{n}\to\infty$ as $n\to\infty.$ \vspace{0.05in}\\
    $\beta_v $ & $\lim_{n\to\infty}\frac{\beta^n_v}{\alpha(n)}$. \vspace{0.05in}\\
$\gamma_v$ & $\lim_{n\to\infty}\frac{\gamma^n_v}{\alpha(n)}$.\vspace{0.05in}\\
 $\hat\beta_v $ & $\lim_{n\to\infty}\sqrt{n}(\frac{\beta^n_v}{\alpha(n)}-\beta_v)$.\vspace{0.05in}\\
    $\hat\gamma_v$ & $\lim_{n\to\infty}\sqrt{n}(\frac{\gamma^n_v}{\alpha(n)}-\gamma_v)$.\\
\hline
   \end{tabular}
   \caption{Parameters}
   \label{tab:booktabs}
\end{table}

\begin{table}[H]
   \centering
   \begin{tabular}{cp{4.3in}}
  Variable/ Process	&	Definition \\
\hline\\
  $S_h^n(t)$ & Number of susceptible hosts at time $t$ when the total host population size is $n$.\vspace{0.05in}\\
  $I_h^n(t)$ & Number of infectious hosts at time $t$ when the total host population size is $n$.\vspace{0.05in}\\
  $S_v^n(t)$ & Number of susceptible vectors at time $t$ when the total host population size is $n$.\vspace{0.05in}\\
  $I_v^n(t)$ & Number of infectious vectors at time $t$ when the total host population size is $n$.\vspace{0.05in}\\
  \hline \\
$\bars_h^n(t) = \frac{S_h^n(t)}{n}$ & Ratio of susceptible hosts at time $t$ to the host population size.\vspace{0.05in}\\
  $\bari_h^n(t)= \frac{I_h^n(t)}{n}$ & Ratio of infectious hosts at time $t$ to the host population size.\vspace{0.05in}\\
  $\bars_v^n(t)= \frac{S_v^n(t)}{n}$ & Ratio of susceptible vectors at time $t$ to the host population size.\vspace{0.05in}\\
  $\bari_v^n(t)= \frac{I_v^n(t)}{n}$ & Ratio of infectious vectors at time $t$ to the host population size.\vspace{0.05in}\\
  $s_h(t)$ & Weak limit of $\bars_h^n(t)$ as the {host population size $n\to\infty$}.\vspace{0.05in}\\
  $i_h^n(t)$ & Weak limit of $\bari_h^n(t)$ as the host population size $n\to\infty$.\vspace{0.05in}\\
  $s_v^n(t)$ & Weak limit of $\bars_v^n(t)$ as the host population size $n\to\infty$.\vspace{0.05in}\\
  $i_v^n(t)$ & Weak limit of $\bari_v^n(t)$ as the host population size$n\to\infty$.\vspace{0.05in}\\
  \hline\\
  $\hats_h^n(t) = \sqrt{n}(\bars_h^n(t) - s_h(t))$ & Scaled deviation of $\bars^n_h(t)$ from its limit $s_h(t)$.\vspace{0.05in}\\
  $\hati_h^n(t)= \sqrt{n}(\bari_h^n(t) - i_h(t))$ & Scaled deviation of $\bari^n_h(t)$ from its limit $i_h(t)$.\vspace{0.05in}\\
  $\hats_v^n(t)= \sqrt{n}(\bars_v^n(t) - s_v(t))$ & Scaled deviation of $\bars^n_v(t)$ from its limit $s_v(t)$.\vspace{0.05in}\\
  $\hati_v^n(t)= \sqrt{n}(\bari_v^n(t) - i_v(t))$ & Scaled deviation of $\bari^n_v(t)$ from its limit $i_v(t)$.\vspace{0.05in}\\
  $\hats_h(t)$ & Weak limit of $\hats_h^n(t)$ as the host population {size} $n\to\infty$.\vspace{0.05in}\\
  $\hati_h(t)$ & Weak limit of $\hati_h^n(t)$ as the host population {size} $n\to\infty$.\vspace{0.05in}\\
  $\hats_v(t)$ & Weak limit of $\hats_v^n(t)$ as the host population {size} $n\to\infty$.\vspace{0.05in}\\
  $\hati_v(t)$ &Weak limit of $\hati_v^n(t)$ as the host population {size} $n\to\infty$.\vspace{0.05in}\\
  \hline\\
  $N_i^n$; $i=1,2,3,4$ & Independent rate-$1$ Poisson processes when the total host population {size} is $n$.\vspace{0.05in} \\
  \hline
   \end{tabular}
   \caption{Processes}
   \label{tab:booktabs}
\end{table}

\end{document}